\documentclass{llncs}
\usepackage[T1]{fontenc}
\usepackage[utf8]{inputenc}
\usepackage{lmodern}
\usepackage{amsmath}
\usepackage{amssymb}
\usepackage{colonequals}
\usepackage{listings}
\usepackage{graphicx}
\usepackage[export]{adjustbox}
\usepackage[]{algorithm}
\usepackage[noend]{algorithmic}

\usepackage{fullpage}
\usepackage{acro}
\usepackage{xspace}
\usepackage{needspace}
\usepackage[inline]{enumitem}
\usepackage[table]{xcolor}
\usepackage{xifthen}

\usepackage{amsthm}
\usepackage{tikz}
\usetikzlibrary{positioning,arrows,backgrounds,automata,decorations.shapes,decorations.pathmorphing,decorations.markings,decorations.text,fit}
\tikzset{
  initial text=$ $ 
}
\usepackage{booktabs} 
\usepackage{todonotes}

\let\IDeclareAcronym\DeclareAcronym
\renewcommand{\DeclareAcronym}[2]{%
 \IDeclareAcronym{#1}{%
  #2,foreign-plural={} 
  }
}
\DeclareAcronym{aa}{
  short = AA,
  long = Acceptance Accuracy,
  long-indefinite = an
}
\DeclareAcronym{adas}{
  short = ADAS,
  long = Advanced Driver-Assistance Systems,
  short-indefinite = an,
  long-indefinite = an
}
\DeclareAcronym{acc}{
  short = ACC,
  long = Adaptive Cruise Control,
  short-indefinite = an,
  long-indefinite = an
}
\DeclareAcronym{adc}{
  short = ADC,
  long  = Analog to Digital Converter,
  short-indefinite = an,
  long-indefinite = an
}
\DeclareAcronym{afr}{
  short = \AFR{},
  long = Always Finitely Refutable,
  short-indefinite = an,
  long-indefinite = an
}
\DeclareAcronym{afs}{
  short = \AFS{},
  long = Always Finitely Satisfiable,
  short-indefinite = an,
  long-indefinite = an
}
\DeclareAcronym{ai}{
  short = AI,
  long  = artificial intelligence
}
\DeclareAcronym{ais}{
  short = AIS,
  long = Active Instance Stack,
  long-indefinite = an
}
\DeclareAcronym{aml}{
  short = AML,
  long  = the ASPEN Modeling Language
}
\DeclareAcronym{anml}{
  short = ANML,
  long  = the Action Notation Modeling Language
}
\DeclareAcronym{anmlite}{
  short = ANMLite,
  long  = the Action Notation Modeling Language Light
}
\DeclareAcronym{api}{
  short = API,
  long  = Application Programming Interface,
  short-indefinite = an,
  long-indefinite = an
}
\DeclareAcronym{atam}{
  short = ATAM,
  long  = Architecture Tradeoff Analysis Method,
  short-indefinite = an,
  long-indefinite = an
}
\DeclareAcronym{atl}{
  short = ATL,
  long = Allen's Temporal Logic,
  short-indefinite = an,
  long-indefinite = an
}
\DeclareAcronym{bide}{
  short = BIDE,
  long = BI-Directional Extension
}
\DeclareAcronym{bji}{
  short = BJI,
  long = between-job inter-arrival time
}
\DeclareAcronym{can}{
  short = CAN,
  long  = Controller Area Network
}
\DeclareAcronym{capec}{
  short = CAPEC,
  long  = Common Attack Pattern Enumeration and Classification
}
\DeclareAcronym{cep}{
  short = CEP,
  long  = Complex Event Processing
}
\DeclareAcronym{cfc}{
  short = CFC,
  long = Computation of Fixed-point Complexity
}
\DeclareAcronym{cfsm}{
  short = CFSM,
  long = communicating finite state machine
}
\DeclareAcronym{cminwin}{
  short = CMINWIN,
  long  = Number of Minimal Windows of Occurrence
}
\DeclareAcronym{cpu}{
  short = CPU,
  long  = Central Processing Unit
}
\DeclareAcronym{csp}{
  short = CSP,
  long = Communicating Sequential Processes
}
\DeclareAcronym{csv}{
  short = CSV,
  long  = Comma Separated Value
}
\DeclareAcronym{dfa}{
  short = DFA,
  long = Deterministic Finite Automaton,
  long-plural-form = Deterministic Finite Automata
}
\DeclareAcronym{dft}{
  short = DFT,
  long  = Discrete Fourier Transform
}
\DeclareAcronym{dsl}{
  short = DSL,
  long = domain-specific language
}
\DeclareAcronym{dsms}{
  short = DSMS,
  long  = Data Stream Management System
}
\DeclareAcronym{dut}{
  short = DUT,
  long = device-under-test
}
\DeclareAcronym{ecu}{
  short = ECU,
  long  = Electronic Control Unit,
  short-indefinite = an,
  long-indefinite = an
}
\DeclareAcronym{edf}{
  short = EDF,
  long = Earliest Deadline First
}
\DeclareAcronym{egads}{
  short = EGADS,
  long  = Extendible and Generic Anomaly Detection System,
  short-indefinite = an,
  long-indefinite = an
}
\DeclareAcronym{epl}{
  short = EPL,
  long = Event Processing Language,
  short-indefinite = an,
  long-indefinite = an
}
\DeclareAcronym{esl}{
  short = ESL,
  long  = Expressive Stream Language,
  short-indefinite = an,
  long-indefinite = an
}
\DeclareAcronym{esql}{
  short = ESQL,
  long  = Event Stream Query Language,
  short-indefinite = an,
  long-indefinite = an
}
\DeclareAcronym{etm}{
  short = ETM,
  long = Embedded Trace Macrocell,
  short-indefinite = an,
  long-indefinite = an
}
\DeclareAcronym{evr}{
  short = EVR,
  long  = EVent Report
}
\DeclareAcronym{fa}{
  short = FA,
  long = Finite Automaton,
  long-plural-form = Finite Automata,
  short-indefinite = an
}
\DeclareAcronym{fcda}{
  short = FCDA,
  long = Forward Collision Detection and Avoidance,
  short-indefinite = an
}
\DeclareAcronym{fifo}{
  short = FIFO,
  long = first-in\, first-out
}
\DeclareAcronym{fltl}{
  short = FLTL,
  long  = \ac{ltl} on Finite Traces,
  short-indefinite = an
}
\DeclareAcronym{fr}{
  short = FR,
  long = Finitely Refutable,
  short-indefinite = an
}
\DeclareAcronym{fs}{
  short = FS,
  long = Finitely Satisfiable,
  short-indefinite = an
}
\DeclareAcronym{fsm}{
  short = FSM,
  long = finite state machine,
  short-indefinite = an
}
\DeclareAcronym{gcc}{
  short = GCC,
  long = the GNU Compiler Collection
}
\DeclareAcronym{gil}{
  short = GIL,
  long = Graphical Interval Logic
}
\DeclareAcronym{gpl}{
  short = GPLv3,
  long  = Gnu Public License version 3
}
\DeclareAcronym{gps}{
  short = GPS,
  long  = Global Positioning System
}
\DeclareAcronym{gui}{
  short = GUI,
  long  = Graphical User Interface
}
\DeclareAcronym{gsql}{
  short = GSQL,
  long  = Gigascope Query Language
}
\DeclareAcronym{hdlc}{
  short = HDLC,
  long = High-level Data Link Control
}
\DeclareAcronym{hmm}{
  short = HMM,
  long = Hidden Markov Model
}
\DeclareAcronym{hpc}{
  short = HPC,
  long  = High Performance Computing,
  short-indefinite = an
}
\DeclareAcronym{hs}{
  short = HS,
  long  = Halpern and Shoham's modal logic of intervals,
  short-indefinite = an
}
\DeclareAcronym{ids}{
  short = IDS,
  long  = Intrusion Detection System,
  short-indefinite = an,
  long-indefinite = an
}
\DeclareAcronym{ifp}{
  short = IFP,
  long  = Information Flow Processing,
  short-indefinite = an,
  long-indefinite = an
}
\DeclareAcronym{iji}{
  short = IJI,
  long = intra-job inter-arrival time,
  short-indefinite = an,
  long-indefinite = an
}
\DeclareAcronym{imu}{
  short = IMU,
  long  = Inertial Measurement Unit,
  short-indefinite = an,
  long-indefinite = an
}
\DeclareAcronym{ip}{
  short = IP,
  long  = Internet Protocol,
  short-indefinite = an,
  long-indefinite = an
}
\DeclareAcronym{isr}{
  short = ISR,
  long = interrupt service routine,
  short-indefinite = an,
  long-indefinite = an
}
\DeclareAcronym{itl}{
  short = ITL,
  long = Interval Temporal Logic,
  short-indefinite = an,
  long-indefinite = an
}
\DeclareAcronym{jit}{
  short = JIT,
  long  = Just-In-Time
}
\DeclareAcronym{jpl}{
  short = JPL,
  long  = Jet Propulsion Laboratory
}
\DeclareAcronym{json}{
  short = JSON,
  long  = JavaScript Object Notation
}
\DeclareAcronym{jvm}{
  short = JVM,
  long  = Java Virtual Machine
}
\DeclareAcronym{knn}{
  short = k-NN,
  long  = k-Nearest Neighbors
}
\DeclareAcronym{lars}{
  short = LARS,
  long = Laboratory for Reliable Software
}
\DeclareAcronym{lanl}{
  short = LANL,
  long  = the Los Alamos National Laboratory
}
\DeclareAcronym{lansce}{
  short = LANSCE,
  long  = Los Alamos Neutron Science Center
}
\DeclareAcronym{lcfsm}{
  short = LCFSM,
  long = lossy communicating finite state machine,
  short-indefinite = an
}
\DeclareAcronym{ldw}{
  short = LDW,
  long = Lane Departure Warning,
  short-indefinite = an
}
\DeclareAcronym{lka}{
  short = LKA,
  long = Lane Keeping Assistance,
  short-indefinite = an
}
\DeclareAcronym{loda}{
  short = LODA,
  long  = Lightweight Online Detector of Anomalies,
  short-indefinite = an,
  long-indefinite = an
}
\DeclareAcronym{ltl}{
  short = LTL,
  long  = Linear Temporal Logic,
  short-indefinite = an
}
\DeclareAcronym{ltl3}{
  short = LTL$_3$,
  long  = Three-value \ac{ltl},
  short-indefinite = an
}
\DeclareAcronym{ltl4}{
  short = LTL$_4$,
  long  = four-value \ac{ltl},
  short-indefinite = an
}
\DeclareAcronym{mcu}{
  short = MCU,
  long  = Micro-controller Unit,
  short-indefinite = an
}
\DeclareAcronym{milprit}{
  short = MILPRIT,
  long = Mining Interval Logic Patterns with Regular expressIons consTraints
}
\DeclareAcronym{msl}{
  short = MSL,
  long = Mars Science Laboratory,
  short-indefinite = an
}
\DeclareAcronym{mtl}{
  short = MTL,
  long = Metric Temporal Logic,
  short-indefinite = an
}
\DeclareAcronym{nasa}{
  short = NASA,
  long  = the National Aeronautics and Space Administration
}
\DeclareAcronym{nba}{
  short = NBA,
  long = Non-deterministic B\"uchi Automaton,
  long-plural-form = Non-deterministic B\"uchi Automata,
  short-indefinite = an
}
\DeclareAcronym{nddl}{
  short = NDDL,
  long  = the EUROPA Modeling Language
}
\DeclareAcronym{nfa}{
  short = NFA,
  long = Non-deterministic Finite Automaton,
  long-plural-form = Non-deterministic Finite Automata,
  short-indefinite = an
}
\DeclareAcronym{nfr}{
  short = \NFR{},
  long = Never Finitely Refutable,
  short-indefinite = an
}
\DeclareAcronym{nfs}{
  short = \NFS{},
  long = Never Finitely Satisfiable,
  short-indefinite = an
}
\DeclareAcronym{os}{
  short = OS,
  long = Operating System,
  short-indefinite = an
}
\DeclareAcronym{pc}{
  short = PC,
  long  = personal computer
}
\DeclareAcronym{pddl}{
  short = PDDL,
  long  = the Planning Domain Definition Language
}
\DeclareAcronym{pid}{
  short = PID,
  long  = Proportional-Integral-Derivative
}
\DeclareAcronym{pocs}{
  short = POCS,
  long  = projections onto convex sets
}
\DeclareAcronym{poset}{
  short = POSET,
  long = Partially-Ordered Set
}
\DeclareAcronym{psdst}{
  short = PSDST,
  long = Propositional Streaming Data-String Transducer,
  short-indefinite = a
}
\DeclareAcronym{psl}{
  short = PSL,
  long = Property Specification Language
}
\DeclareAcronym{ptl}{
  short = PTL,
  long  = Propositional Temporal Logic
}
\DeclareAcronym{ptp}{
  short = PTP,
  long  = Precision Time Protocol
}
\DeclareAcronym{qbf}{
  short = QBF,
  long = Quantified Boolean Formula,
  long-plural-form = Quantified Boolean Formulae
}
\DeclareAcronym{qea}{
  short = QEA,
  long = Quantified Event Automata
}
\DeclareAcronym{ra}{
  short = RA,
  long = Rejection Accuracy,
  short-indefinite = an
}
\DeclareAcronym{rems}{
  short = REMS,
  long  = Rover Environmental Monitoring Station 
}
\DeclareAcronym{resp}{
  short = RESP,
  long  = the REdis Serialization Protocol,
  short-indefinite = an
}
\DeclareAcronym{rm}{
  short = RM,
  long = Rate Monotonic,
  short-indefinite = an
}
\DeclareAcronym{rnn}{
  short = RNN,
  long  = Recurrent Neural Network,
  short-indefinite = an
}
\DeclareAcronym{roc}{
  short = ROC,
  long  = Receiver Operating Characteristic,
  short-indefinite = an
}
\DeclareAcronym{ros}{
  short = ROS,
  long  = Robot Operating System
}
\DeclareAcronym{rpn}{
  short = RPN,
  long  = Reverse Polish Notation,
  short-indefinite = an
}
\DeclareAcronym{rv}{
  short = RV,
  long = Runtime Verification,
  short-indefinite = an
}
\DeclareAcronym{rvltl}{
  short = RV-LTL,
  long  = Runtime Verification \ac{ltl},
  short-indefinite = an
}
\DeclareAcronym{saam}{
  short = SAAM,
  long = Software Architecture Analysis Method,
  short-indefinite = an
}
\DeclareAcronym{sax}{
  short = SAX,
  long  = Symbolic Aggregate approXimation
}
\DeclareAcronym{sclwnrf}{
  short = LANL,
  long  = System Call Logs with Natural Random Faults
}
\DeclareAcronym{sdk}{
  short = SDK,
  long = Software Development Kit
}
\DeclareAcronym{sdl}{
  short = SDL,
  long = Specification and Description Language,
  short-indefinite = an
}
\DeclareAcronym{sdst}{
  short = SDST,
  long = Streaming Data-String Transducer,
  short-indefinite = an
}
\DeclareAcronym{sfr}{
  short = \SFR{},
  long = Sometimes Finitely Refutable,
  short-indefinite = an
}
\DeclareAcronym{sfs}{
  short = \SFS{},
  long = Sometimes Finitely Satisfiable,
  short-indefinite = an
}
\DeclareAcronym{sft}{
  short = SFT,
  long = Symbolic Finite-state Transducer,
  short-indefinite = an
}
\DeclareAcronym{smm}{
  short = SMM,
  long = Symbolic Mealy Machine,
  short-indefinite = an
}
\DeclareAcronym{soc}{
  short = SoC,
  long  = System-on-a-Chip,
  short-indefinite = an,
  long-plural-form = Systems-on-Chips
}
\DeclareAcronym{spade}{
  short = SPACE,
  long = Sequential PAttern Discovery using Equivalence classes
}
\DeclareAcronym{spam}{
  short = SPAM,
  long = Sequential PAttern Mining
}
\DeclareAcronym{sql}{
  short = SQL,
  long  = Structured Query Language
}
\DeclareAcronym{srv}{
  short = SRV,
  long = Stream Runtime Verification,
  short-indefinite = an
}
\DeclareAcronym{ssps}{
  short = SSPS,
  long  = Sequential Sense-Process-Send
}
\DeclareAcronym{stdev}{
  short = stdev,
  long  = standard deviation
}
\DeclareAcronym{stl}{
  short = STL,
  long = Signal Temporal Logic,
  short-indefinite = an
}
\DeclareAcronym{svpa}{
  short = SVPA,
  long = Symbolic Visibly Pushdown Automaton,
  long-plural-form = Symbolic Visibly Pushdown Automata,
  short-indefinite = an
}
\DeclareAcronym{tltl}{
  short = TLTL,
  long  = Timed \ac{ltl}
}
\DeclareAcronym{tksr}{
  short = TKSR,
  long = Time Series Knowledge Representation
}
\DeclareAcronym{tre}{
  short = TRE,
  long = Timed Regular Expression
}
\DeclareAcronym{udp}{
  short = UDP,
  long = User Datagram Protocol
}
\DeclareAcronym{uhml}{
  short = $\mu$HML,
  long = $\mu$-Hennessy-Milner Logic,
  short-indefinite = a
}
\DeclareAcronym{uml}{
  short = UML,
  long  = Unified Modeling Language
}
\DeclareAcronym{wcet}{
  short = WCET,
  long = Worst-Case Execution Time
}
\DeclareAcronym{wsdl}{
  short = WSDL,
  long = Web Services Description Language
}
\DeclareAcronym{xml}{
  short = XML,
  long  = Extensible Markup Language,
  short-indefinite = an
}

\newcommand{\set}[1]{\{#1\}}

\newcommand{\true}{\mathit{true}}
\newcommand{\false}{\mathit{false}}

\newcommand{\cleft}{\ell}
\newcommand{\cright}{r}

\newcommand{\pspace}{\textsc{PSpace}\xspace}
\newcommand{\ptime}{\textsc{PTime}\xspace}
\newcommand{\aptime}{\textsc{APTime}\xspace}
\newcommand{\apspace}{\textsc{APSpace}\xspace}
\newcommand{\exptime}{\textsc{ExpTime}\xspace}
\newcommand{\expspace}{\textsc{ExpSpace}\xspace}
\newcommand{\nexptime}{\textsc{NExpTime}\xspace}
\newcommand{\aexptimepoly}{\textsc{AExpTime}(poly)\xspace}

\newcommand{\tm}{\mathcal{M}}

\newcommand{\nfer}{{\tt nfer}\xspace}
\newcommand{\Nfer}{{\tt Nfer}\xspace}

\newcommand{\infer}{{\tt inc-}\nfer}
\newcommand{\Infer}{{\tt Inc-}\nfer}

\newcommand{\extnfer}{{\tt ext-}\nfer}

\newcommand{\etal}{et al.\xspace}

\renewcommand{\mid}{:}

\newcommand{\natty}{\mathbb{N}}
\newcommand{\boolty}{\mathbb{B}}

\newcommand{\powerset}[1]{2^{#1}}
\newcommand*{\Cdot}{\raisebox{-0.25ex}{\scalebox{1.6}{.}}}
\newcommand{\where}[1][1]{\Cdot\ifnum #1=0{}\else{\foreach \n in {1,...,#1}{\ }}\fi} 

\newcommand{\Ident}{\mathcal{I}}

\newcommand{\mapty}{\mathbb{M}}
\newcommand{\emptymap}{\{\ \}}
\makeatletter
\newcommand{\pto}{}
\newcommand{\pgets}{}

\DeclareRobustCommand{\pto}{\mathrel{\mathpalette\p@to@gets\to}}
\DeclareRobustCommand{\pgets}{\mathrel{\mathpalette\p@to@gets\gets}}

\newcommand{\p@to@gets}[2]{%
  \ooalign{\hidewidth$\m@th#1\mapstochar\mkern5mu$\hidewidth\cr$\m@th#1\to$\cr}%
}
\makeatother

\newcommand{\Idkw}{\textit{id}}
\newcommand{\Startkw}{\textit{start}}
\newcommand{\Endkw}{\textit{end}}
\newcommand{\Mapkw}{\textit{map}}
\newcommand{\Idof}[1]{\Idkw(#1)}
\newcommand{\Startof}[1]{\Startkw(#1)}
\newcommand{\Endof}[1]{\Endkw(#1)}
\newcommand{\Mapof}[1]{\Mapkw(#1)}

\newcommand{\clockty}{\natty}

\newcommand{\target}{\eta_{T}\xspace}

\newcommand{\beforekw}{{\bf before}}
\newcommand{\duringkw}{{\bf during}}
\newcommand{\slicekw}{{\bf slice}}
\newcommand{\meetkw}{{\bf meet}}
\newcommand{\coincidekw}{{\bf coincide}}
\newcommand{\startkw}{{\bf start}}
\newcommand{\finishkw}{{\bf finish}}
\newcommand{\overlapkw}{{\bf overlap}}

\newcommand{\unlesskw}{{\bf unless}}
\newcommand{\containkw}{{\bf contain}}
\newcommand{\followkw}{{\bf follow}}
\newcommand{\afterkw}{{\bf after}}

\newcommand{\before}[2]{#1 \ \beforekw\  #2}

\newcommand{\meet}[2]{#1 \ \meetkw\  #2}

\newcommand{\follow}[2]{#1 \ \unlesskw \ \followkw\ #2}

\newcommand{\nferrule}[6]{#1 \leftarrow #2\ #3\ #4 \textbf{ where } #5 \textbf{ map } #6}
\newcommand{\coinciderule}[5]{\nferrule{#1}{#2}{\coincidekw{}}{#3}{#4}{#5}}

\makeatletter
\newcommand{\shorteq}{%
  \settowidth{\@tempdima}{n}
  \resizebox{\@tempdima}{\height}{=}%
}
\makeatother

\newcommand{\eventty}{\mathbb{E}}
\newcommand{\tracety}{\eventty^*}
\newcommand{\intervalty}{\mathbb{I}}

\newcommand{\event}{\varepsilon}

\newcommand{\pool}{\pi}
\newcommand{\selection}{\partial}

\newcommand{\minimality}{\textbf{minimality}\,}

\newcommand{\interp}[2]{R[ \! [ #1 ] \! ] \: #2}

\newcommand{\interpmin}[2]{R_{\text{\tiny{min}}}[ \! [ #1 ] \! ] \: #2}
\newcommand{\seminterp}[2]{S[ \! [ #1 ] \! ] \: #2}
\newcommand{\traceinterpX}[4][]{%
\if\relax\detokenize{#4}\relax
  \if\relax\detokenize{#3}\relax
    T[ \! [ #2 ] \! ]%
  \else
  T\!{\text{\tiny#3}}[ \! [ #2 ] \! ]%
  \fi
\else
  \if\relax\detokenize{#3}\relax
  T\!{\overset{{\scriptscriptstyle#4}}{\text{\tiny#3}}}[ \! [ #2 ] \! ]%
  \else
  T\!{\overset{{\scriptscriptstyle#4}}{\text{\tiny#3}}}[ \! [ #2 ] \! ]%
  \fi
\fi%
\ifthenelse{\isempty{#1}}{}{\: #1}}

\newcommand{\traceinterpi}[3][]{\traceinterpX[#3]{#2}{inc}{#1}}
\newcommand{\traceinterpf}[3][]{\traceinterpX[#3]{#2}{full}{#1}}
\newcommand{\traceinterpe}[3][]{\traceinterpX[#3]{#2}{ext}{#1}}

\newcommand{\poolty}{\mathbb{P}}
\newcommand{\rulety}{\Delta}

\newcommand{\initf}[1]{\textit{init}(#1)}

\lstdefinelanguage{semantics}{
  morekeywords=
    {if,then,else,let,in,least,such,that,unless,it,exists}, 
  otherkeywords={}, 
  literate=
    {[:}{{$[\![$}}2
    {:]}{{$]\!]$}}2
    {Pool}{{$\poolty$}}2
    {Interval}{{$\intervalty$}}2
    {Map}{{$\mapty$}}2
    {Clock}{{$\clockty$}}4
    {Trace}{{$\tracety$}}2
    {Nats}{{$\natty$}}2
    {emptyset}{{$\varnothing$}}2
    {forall}{{$\forall$}}2
    {overlineforall}{{$\overline{\forall}$}}2
    {not}{{$\neg$}}2
    {isin}{{$\in$}}1
    {isnotin}{{$\notin$}}2
    {union}{{$\cup$}}2
    {UNION}{{$\bigcup$}}2
    {intersect}{{$\cap$}}2
    {superset}{{$\supset$}}2
    {subseteq}{{$\subseteq$}}2
    {max}{{$\max$}}2
    {min}{{$\min$}}2
    {Exists}{{$\exists$}}2
    {selection}{{$\selection$}}2
    {minimality}{{\minimality}}8
    {:-}{{$\where$}}2  
    {&&}{{$\wedge$}}2
    {||}{{$\vee$}}2
    {!=}{{$\neq$}}2
    {>=}{{$\geq$}}2
    {=<}{{$\leq$}}2
    {<.}{{$\langle$}}2
    {.>}{{$\rangle$}}2
    {...}{{$\cdots$}}2
    {oplus}{{$\oplus$}}2
    {bottom}{{$\bot$}}2
    {emptymap}{{$\emptymap$}}2
    {ominus}{{$\ominus$}}2
    {::=}{{$\doteq$}}2
    {<-}{{$\leftarrow$}}2
    {->}{{$\rightarrow$}}2
    {<=}{{$\Leftarrow$}}2
    {=>}{{$\Rightarrow$}}2
    {!=>}{{$\not\Rightarrow$}}2
    {|->}{{$\mapsto$}}2
    {Mapof}{{$\Mapkw$}}3
    {Idof}{{$\Idkw$}}2
    {Startof}{{$\Startkw$}}4
    {Endof}{{$\Endkw$}}2
    {mapkw}{{\bf map}}4
    {wherekw}{{\bf where}}6
    {RuleSem}{{$R$}}2
    {RuleSel}{{$R_{\selection}$}}2
    {RuleMin}{{$R_{\scriptscriptstyle{min}}$}}2
    {SpecSem}{{$S$}}2
    {TraceSemI}{{$T\!{\text{\tiny{inc}}}$}}2
    {TraceSemF}{{$T\!{\text{\tiny{full}}}$}}1
    {TraceSemE}{{$T\!{\text{\tiny{ext}}}$}}1
    {Lambda}{$\Lambda$}2
    {Delta}{$\Delta$}2
    {Phi}{$\Phi$}2
    {Psi}{$\Psi$}2
    {Upsilon}{$\Upsilon$}2
    {pi}{$\pi$}2
    {pi1}{$\pi_1$}2
    {pi2}{$\pi_2$}2
    {piI}{$\pi_i$}2
    {piJ}{$\pi_{i+1}$}4
    {piN}{$\pi_n$}2
    {piprime}{$\pi^{\prime}$}2
    {eta}{$\eta$}1
    {etaP}{$\eta^{\prime}$}1
    {eta1}{$\eta_1$}2
    {eta2}{$\eta_2$}2  
    {eta3}{$\eta_3$}2
    {tau}{{$\tau$}}2
    {eps}{{$\event$}}2
    {eps1}{{$\event_1$}}2
    {eps2}{{$\event_2$}}2
    {eff}{{$f$}}2
    {gee}{{$g$}}2
    {s0}{{$s$}}1
    {e0}{{$e$}}1
    {J0}{{$J$}}1
    {K0}{{$K$}}1
    {L0}{{$L$}}1
    {M0}{{$M$}}1
    {MP}{{$M^{\prime}$}}1
    {s1}{{$s_1$}}2
    {s2}{{$s_2$}}2
    {s3}{{$s_3$}}2
    {sP}{{$s^{\prime}$}}1
    {e1}{{$e_1$}}2
    {e2}{{$e_2$}}2
    {e3}{{$e_3$}}2
    {eP}{{$e^{\prime}$}}1
    {i0}{{$i$}}1
    {i1}{{$i_{1}$}}1
    {i2}{{$i_{2}$}}1
    {init}{{$init$}}2
    {topsort}{{\textit{topsort}}}4
    {lambda1}{{$\lambda_1$}}2
    {lambdaN}{{$\lambda_n$}}2
    {LambdaSet}{{$\powerset{\Lambda}$}}2
    {rule}{{$\delta$}}2
    {delta1}{{$\delta_1$}}2
    {delta2}{{$\delta_2$}}2
    {deltaN}{{$\delta_n$}}2
    {Di}{{$D_i$}}2
    {D1}{{$D_1$}}2
    {Dl}{{$D_\ell$}}2
    {,,,}{{$,\ldots,$}}3
    {Comps}{{$\mathcal{D}$}}2
    {SIZEl}{{$_{\ell+1}$}}2
    {ENN}{$n$}1
    {DeltaList}{{$\Delta^{*}$}}2
    {RuleList}{{$\mathbf{d}$}}2
    {AllI}{{$_{i > 0}$}}3,
  sensitive=true, 
  morecomment=[l]{//}, 
  morecomment=[s]{/*}{*/},
  escapeinside={(*}{*)},
  stringstyle=\ttfamily,
  aboveskip=5mm,
  belowskip=5mm,
  showstringspaces=false,
  columns=flexible,
  morestring=[b]",
  numberstyle=\scriptsize,
  moredelim=[is][\em]{@}{@}
}

\newcommand{\stopp}{\texttt{STOP}}
\newcommand{\inc}[1]{\texttt{INC(X}_{#1}\texttt{)}}
\newcommand{\dec}[1]{\texttt{DEC(X}_{#1}\texttt{)}}
\newcommand{\ite}[2]{\texttt{IF X}_{#1}\texttt{=0 GOTO }#2}
\newcommand{\instr}{\texttt{I}}

\newcommand{\semantics}{\lstset{language=semantics,backgroundcolor=\color{white}, frame=none,basicstyle={\normalsize},aboveskip=.25em,belowskip=.25em}}
\newcommand{\isem}{\semantics\lstinline}

\begin{document}

\title{The Complexity of Evaluating nfer}

\author{
Sean Kauffman
\and
Martin Zimmermann
}
\institute{
Aalborg University, Denmark\\
\email{\{seank,mzi\}@cs.aau.dk}
}

\maketitle \thispagestyle{plain} \pagestyle{plain}

\begin{abstract}

\Nfer is a rule-based language for abstracting event streams into a hierarchy of intervals with data.
\Nfer has multiple implementations and has been applied in the analysis of spacecraft telemetry and autonomous vehicle logs.
This work provides the first complexity analysis of \nfer evaluation, i.e., the problem of deciding whether a given interval is generated by applying rules.

We show that the full \nfer language is undecidable and that this depends on both recursion in the rules and an infinite data domain.
By restricting either or both of those capabilities, we obtain tight decidability results.
We also examine the impact on complexity of exclusive rules and minimality.
For the most practical case, which is minimality with finite data, we provide a polynomial-time algorithm.

\end{abstract}

\section{Introduction}
\label{sec:introduction}
\Nfer is a rule-based language for event stream analysis, developed with scientists from \ac{nasa}'s \ac{jpl} to analyze telemetry from spacecraft~\cite{kauffman2016towards,kauffman2016nfer,kauffman2017inferring}.
\Nfer rules calculate data over periods of time called intervals.
\Nfer compares and combines these intervals to form a hierarchy of abstractions that is easier for humans and machines to comprehend than a trace of discrete events.
This differs from traditional \ac{rv} which computes language inclusion and returns verdicts.
The equivalent problem for \nfer, called the evaluation problem, is to determine if an interval will be present in \nfer's output given a list of rules and an input trace.

The \nfer syntax is based on \ac{atl}~\cite{allen1983maintaining} and is designed for simplicity and brevity in many contexts.
When it was originally introduced, \nfer was used to find false positives among warning messages from the \ac{msl}, i.e., the Curiosity rover, at \ac{jpl}~\cite{kauffman2016nfer}.
Researchers found the language to be much more concise than the ad hoc Python scripts in common use.
\Nfer has also been deployed to capture disagreements between parallel \ac{pid} controllers in an embedded system ionizing radiation experiment~\cite{narayan2017system,kauffman2017inferring} and to locate unstable gear shifts in an autonomous vehicle~\cite{kauffman2020palisade}.

\Nfer is expressive enough for many applications and termination of \nfer has been conjectured to be undecidable~\cite{kauffman2021phd}.
The intuition for \nfer undecidability is that recursion in its rules is possible and the intervals \nfer computes may carry data from an infinite domain.

Despite this expressiveness, \nfer's implementations have been demonstrated to be fast in practice.
Both the C~\cite{nferio} and Scala~\cite{nfer-scala} versions have been compared against tools such as LogFire and Prolog~\cite{kauffman2017inferring}, Siddhi~\cite{kauffman2020palisade}, MonAmi and \texttt{DejaVu}~\cite{havelund2021monami}, and TeSSLa~\cite{kauffman2021nfer} and in every case found to be faster than the alternatives performing the same analysis.
The question remains if \nfer evaluation is indeed undecidable and, if so, if there are useful fragments of \nfer with a tractable evaluation problem.


\paragraph{Our Contribution}



In this work, we determine the complexity of evaluating different fragments of \nfer.
We find that any one of several restrictions on the language permit decidable evaluation and we prove tight bounds for almost all of these fragments.

We begin by defining a natural syntactic fragment of \nfer using only inclusive rules called \infer.
Full \nfer supports a form of negation using what are called exclusive rules, but we show that these are unnecessary to obtain undecidability.
The result relies, instead, on recursion between rules and on intervals carrying data from an infinite domain.
This proves the conjecture mentioned above.

To regain decidability, we then examine language fragments where either or both of these capabilities are restricted.
We prove that, without recursion, \infer evaluation is \nexptime-complete, without infinite data it is \exptime-complete, and without either it is \pspace-complete.

We then introduce exclusive rules and examine the full \nfer language.
It has been openly questioned what effect negation has on the expressiveness of \nfer~\cite{havelund2021monami}.
Of note is that the semantics of exclusive rules restrict the use of recursion.
In particular, to ensure that an \nfer instance with exclusive rules yields a unique output, exclusive rules may not appear within cycles.
\Nfer previously only supported exclusive rules in cycle-free settings, but we extend the semantics here to support exclusive rules alongside cycles in inclusive (i.e., non-exclusive) rules.

We prove that, with finite data, adding exclusive rules has no effect and \nfer evaluation remains \pspace-complete without cycles and \exptime-complete with cycles.
Without cycles and infinite data, however, we prove that the problem is \aexptimepoly-complete.
The remaining combination, i.e., exclusive rules with cycles and infinite data is undecidable, as the problem is already so without exclusive rules.

We go on to examine the effect of minimality on the complexity of \nfer evaluation.
Minimality is a so-called meta-constraint on the results of \nfer that was a primary motivator of \nfer's development, since it was discovered existing tools like Prolog struggled with such meta-constraints~\cite{kauffman2017inferring}.
We show that minimality has a substantial effect on the complexity of \nfer evaluation.
With infinite data, we prove that the problem is in \exptime.
The most common method of using \nfer is with minimality and finite data, however, and we prove evaluation for this configuration is in \ptime.
These results are independent of the use of exclusive rules and cycles. 

Table~\ref{table:allrefs} shows an overview of our results.

\begin{table}[ht]
  \caption{Overview of our complexity results. The asterisk~$^\ast$ signifies that only cycles consisting of inclusive rules are allowed.}
  \label{table:allrefs}
  \centering
  \begin{tabular}{llllll}
  \toprule
Semantics & Cycles & Data & Minimality & Complexity & Reference\\
\midrule
inclusive & yes & infinite & no & undecidable & Theorem \ref{thm:full}\\

inclusive & no  & finite   & no & \pspace-complete&Theorem \ref{thm:inc_free_finite}\\

inclusive & yes  & finite   & no & \exptime-complete&Theorem \ref{thm:inc_finite}\\

inclusive & no  & infinite & no & \nexptime-complete&Theorem~\ref{thm:inc_free_infinite}\\

\midrule

full & no  & finite   & no & \pspace-complete&Theorem \ref{thm:excl_free_finite}\\

full & no  & infinite & no & \aexptimepoly-complete&Theorem \ref{thm:excl_free_infinite}\\

\midrule

extended  & yes$^\ast$& finite   & no & \exptime-complete&Theorem \ref{thm:extended_finite}\\

extended  & yes$^\ast$& infinite   & no & undecidable &Theorem \ref{thm:full}\\

extended  & yes$^\ast$& finite   & yes& \ptime&Theorem \ref{thm:minimality_finite}\\

extended  & yes$^\ast$& infinite & yes& \exptime&Theorem \ref{thm:minimality_infinite}\\
  \end{tabular}
 \end{table}

This work extends the authors' previous paper presented at TASE~2022~\cite{kauffman2022complexity}.
Here, we present the complete proofs for each complexity result.
We also prove a new lower bound for full semantics with infinite data.
Finally, we introduce the \nfer extended semantics and provide complexity results for the relevant fragments under it.

\paragraph{Related Work}

\Nfer is closely related to other classes of declarative programming systems but it differs from them all in several ways.
For example, a rule-based programming system modifies a database of facts~\cite{barringer2011tracecontract,havelund2015logfire}.
Unlike these systems, however, \nfer is monotonic and can only add intervals, not remove them.
\Nfer also resembles \ac{cep} systems where declarative rules are applied to compute information from a trace of events~\cite{chen2000niagara,luckham2008cep,suhothayan2011siddhi}.
\Ac{cep} systems do not usually include explicit notions of time or temporal relationships, though, which are central to \nfer.
In this way, \nfer more closely resembles stream-\ac{rv} systems~\cite{halle2016rv,convent2018tessla,faymonville2019realtime}.
Still, \nfer is differentiated from these systems by its emphasis on temporal intervals and its \ac{atl}-based syntax.


Some research has examined the complexity of logics based on \ac{atl}, specifically \ac{hs}~\cite{halpern1991propositional}.
Montanari~\etal showed that the satisfiability problem for the subset of \ac{hs} consisting of only the temporal operators \emph{begins/begun by} and \emph{meets} is \expspace-complete over the natural numbers~\cite{montanari2010abb}.
Later, they showed that adding the \emph{met by} operator increases the complexity such that the language is only decidable over finite total orders~\cite{montanari2010maximal}.
Aceto~\etal identified the expressive power of all fragments of \ac{hs} over total orders as well as only dense total orders~\cite{aceto2016classification}.
\Nfer is not a modal logic, however, and these complexity results are not relevant to its evaluation problem.

\section{The Inclusive nfer Language}
\label{sec:language}

\semantics

The \nfer language supports two types of rules: inclusive rules and exclusive rules.
This section desribes the \texttt{inclusive-}\nfer formalism, subsequently abbreviated \infer, that supports only inclusive rules.
\Infer is sufficiently expressive to obtain an undecidability result and we find that initially omitting exclusive rules simplifies our presentation.
\Infer is also a natural subset of \nfer that was first introduced in~\cite{kauffman2016nfer}.
It supports many use cases, including the \ac{msl} case-study described above.
The implementation of \nfer written in Scala at \ac{jpl}~\cite{nfer-scala,kauffman2017inferring} also supports only inclusive rules.
%
We expand our analysis to include exclusive rules in Section~\ref{sec:exclusive} (without any cycles) and in Section~\ref{sec:extended} (with cycles in inclusive rules). 
Finally, Section~\ref{sec:minimality} addresses minimality, an important extension of \nfer semantics.
Note that, to improve comprehensibility and simplify later proofs, the semantics presented in this section differs slightly from prior work but these changes do not affect the language capabilities.


\subsection{Preliminary Notation}
\label{paragraph:notation}

We denote the set of nonnegative integers as $\natty$. 
The set of Booleans is given as $\boolty = \{\true, \false\}$.
We fix a finite set~$\Ident$ of identifiers.
$\mapty$ is the set of all of maps, where a map $M \in \mapty$ is a partial function ${M : \Ident \pgets \natty \cup \boolty}$.


An event represents a named state change in an observed system.
An event is a triple $(\eta, t, M)$ where $\eta \in \Ident$ is its identifier, $t \in \clockty$ is the timestamp when it occurred, and $M \in \mapty$ is its map of data.
The set of all events is $\eventty = \Ident \times \clockty \times \mapty$.
A sequence of events $\tau \in \eventty^*$ is called a \emph{trace}.

Intervals represent a named period of state in an observed system.
An interval is a 4-tuple $(\eta, s, e, M)$ where $\eta \in \Ident$ is its identifier, $s, e \in \clockty$ are the starting and ending timestamps where $s \leq e$, and $M \in \mapty$ is its map of data.
The set of all intervals with data is $\intervalty = \Ident \times \clockty \times \clockty \times \mapty$. 
A set of intervals is called a \emph{pool} and the set of all pools is $\poolty = \powerset{\intervalty}$.
We say that an interval~$i = (\eta, s, e, M)$ is labeled by $\eta$.
We define the functions $\Idof{i} = \eta$, $\Startof{i} = s$, $\Endof{i} = e$, and $\Mapof{i} = M$.

\subsection{Syntax}
\label{paragraph:syntax}
Inclusive rules test for the existence of two intervals matching constraints.
When such a pair is found, a new interval is produced with an identifier specified by the rule.
The new interval has timestamps and a map derived by applying functions, specified in the rule, to the matched pair of intervals.
We define the syntax of these rules, including mathematical functions to simplify the presentation, as follows:
\[\nferrule{\eta}{\eta_1}{\oplus}{\eta_2}{\Phi}{\Psi}\]

\noindent
where, $\eta, \eta_1, \eta_2 \in \Ident$ are identifiers, 
$\oplus \in \{\beforekw{},\meetkw{},\duringkw{},\coincidekw{},$ $\startkw{},$ $\finishkw{},\overlapkw{},\slicekw{}\}$ is a \emph{clock predicate} on three intervals (one for each of $\eta,\eta_1,\text{ and } \eta_2$),
$\Phi : \mapty \times \mapty \rightarrow \boolty$ is a \emph{map predicate} taking two maps and returning a Boolean representing satisfaction of a constraint, and 
$\Psi : \mapty \times \mapty \rightarrow \mapty$ is a \emph{map update} taking two maps and returning a map.

We omit the precise syntax for specifying map predicates and updates, but we require that these functions are limited to only simple arithmetic operations.
This matches what is possible using the C \nfer tool~\cite{kauffman2021nfer}.
Specifically, map predicates and map updates must be expressible using the standard mathematical operations $+$, $-$, $\cdot$ (multiplication), $/$ (division), $\bmod$, the comparisons $<,\leq,>,\geq,=$, and the Boolean operators $\wedge, \vee, \neg$.
This limitation excludes exponentiation and any form of recursion in the functions.
Since we do not support real numbers in the theory, division is limited to integer quotients.
These decisions are discussed in Section~\ref{sec:discussion}.

\subsection{Semantics}
\label{paragraph:semantics}

\Infer defines how rules are interpreted to generate pools of intervals from inputs.
The semantics utilizes functions, referenced by the rule syntax, that specify the temporal and data relationships between intervals.
%
The semantics of the \nfer language is defined in three steps: the semantics~$R$ of individual rules on pools, the semantics $S$ of a specification (a list of rules) on pools, and finally the semantics $T$ of a specification on traces of events.

We first define the semantics of inclusive rules with the interpretation function $R$.
Let $\rulety$ be the set of all rules.
Semantic functions are defined using the brackets \isem{[:_:]} around syntax being given semantics.

\begin{lstlisting}
       RuleSem[:_:] : Delta -> Pool -> Pool
       RuleSem[:eta <- eta1 oplus eta2 wherekw Phi mapkw Psi:]pi = 
            { i0 isin  Interval : i1,i2 isin  pi :-
                        Idof(i0) = eta && Idof(i1) = eta1 && Idof(i2) = eta2  &&
                        oplus(i0,i1,i2) && Phi(Mapof(i1),Mapof(i2)) && 
                        Mapof(i0) = Psi(Mapof(i1),Mapof(i2)) }
\end{lstlisting}
In the definition, a new interval $i$ is produced when two existing intervals in $\pool$ match the identifiers $\eta_1$ and $\eta_2$, the temporal constraint $\oplus$, and the map constraint $\Phi$.
$\oplus$ defines the start and end timestamps of $i$ and $\Psi$ defines its map.

The possibilities referenced by $\oplus$ are shown in Figure~\ref{fig:inclusiveops} and formally defined in Table~\ref{table:inclusiveops}.
These clock predicates relate two intervals using the familiar \ac{atl} temporal operators and also specify the start and end timestamps of the produced intervals.
In the figure, the two matched intervals are shown as dark-gray boxes where time flows from left to right and the light-gray box is the produced interval.
Note that the generated interval labeled by $C$ has start and end timestamps inherited from the intervals labeled $A$ and $B$, i.e., no new timestamps are generated by applying a rule.
For example, given intervals $i,i_{1},i_{2}$ where ${\Idof{i} = A}$, ${\Idof{i_{1}} = B}$ and ${\Idof{i_{2}} = C}$, ${A \leftarrow \meet{B}{C}}$ holds when $\Endof{i_{1}} = \Startof{i_{2}}$, $\Startof{i} = \Startof{i_{1}}$, and $\Endof{i} = \Endof{i_{2}}$.

\newcommand{\allen}[1]{{{\small A $\leftarrow$ B \textbf{#1} C}}}
\begin{figure}[t]
\centering
\begin{tabular}{c|c|c|c}
\rowcolor{white}
\allen{before} & \allen{meet} & \allen{during} & \allen{coincide}\\
\includegraphics[height=1cm]{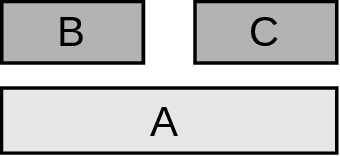} & \includegraphics[height=1cm]{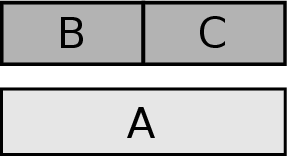} & \includegraphics[height=1.4cm]{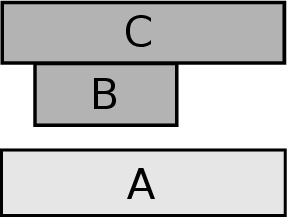} &  \includegraphics[height=1.4cm]{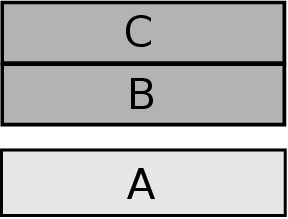}\\
\hline
\rowcolor{white}
\allen{start} & \allen{finish} & \allen{overlap} & \allen{slice} \\
\includegraphics[height=1.4cm]{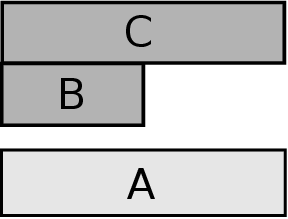} & \includegraphics[height=1.4cm]{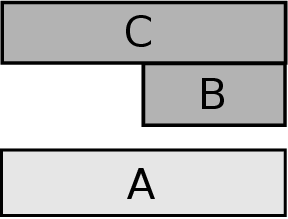} & \includegraphics[height=1.4cm]{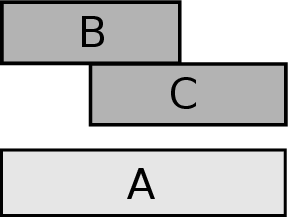} & \includegraphics[height=1.4cm]{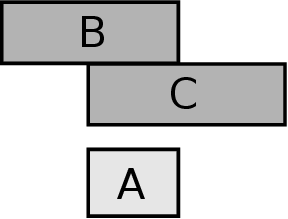} \\
\end{tabular}
\caption{\nfer clock predicates for inclusive rules}
\label{fig:inclusiveops}
\end{figure}

\begin{table}[ht]
  \caption{Formal definition of \nfer clock predicates for inclusive rules}
  \label{table:inclusiveops}
  \centering
  \begin{tabular}{ll}
  \toprule
  Syntax & Definition of $\oplus(i,i_1,i_2)$ \\
  \midrule

  \beforekw{} & $\Endof{i_1} < \Startof{i_2} \wedge \Startof{i} = \Startof{i_1} \wedge \Endof{i} = \Endof{i_2}$\\

\rowcolor{gray!30}  \meetkw{} & $\Endof{i_1} = \Startof{i_2} \wedge \Startof{i} = \Startof{i_1} \wedge \Endof{i} = \Endof{i_2}$\\

    \duringkw{} & $\Startof{i_1} \geq \Startof{i_2} = \Startof{i} \wedge \Endof{i_1} \leq \Endof{i_2} = \Endof{i}$\\

\rowcolor{gray!30}    \coincidekw{} & $\Startof{i_1} = \Startof{i_2} = \Startof{i} \wedge \Endof{i_1} = \Endof{i_2} = \Endof{i}$\\

    \startkw{} & $\Startof{i_1} = \Startof{i_2} = \Startof{i} \wedge \Endof{i} = \max(\Endof{i_1},\Endof{i_2})$\\

\rowcolor{gray!30}  \finishkw{} & $\Startof{i} = \min(\Startof{i_1},\Startof{i_2}) \wedge \Endof{i} = \Endof{i_1} = \Endof{i_2}$\\

  &  $\Startof{i_1} < \Endof{i_2} \wedge \Startof{i_2} < \Endof{i_1}\ \wedge$ \\ 
  \overlapkw{} &   $\Startof{i} = \min(\Startof{i_1},\Startof{i_2})\ \wedge$ \\
   & $\Endof{i} = \max(\Endof{i_1},\Endof{i_2})$\\

\rowcolor{gray!30}&    $\Startof{i_1} < \Endof{i_2} \wedge \Startof{i_2} < \Endof{i_1}\ \wedge$ \\ 
 \rowcolor{gray!30} \slicekw{} &    $\Startof{i} = \max(\Startof{i_1},\Startof{i_2})\ \wedge$ \\
\rowcolor{gray!30} &   $\Endof{i} = \min(\Endof{i_1},\Endof{i_2})$\\

  \end{tabular}
\end{table}

The following one-step interpretation function $S$ defines the semantics of a finite list of rules, also called a specification. 
Given a specification~$\delta_1 \cdots \delta_n \in \rulety^*$ and a pool $\pool \in \poolty$, $\seminterp{\_}$ returns a new pool obtained by recursively applying $\interp{\_}$ to every rule in $\delta_1\cdots \delta_n$ in order, where each is called using the union of $\pool$ with the new intervals returned thus far.
\begin{align*}
&S\;[\![\_]\!]\ :\ \rulety^* \rightarrow \poolty \rightarrow \poolty\\
&S\;[\![\ \delta_1 \cdots \delta_n\ ]\!]\ \pi\ =\ 
\begin{cases}
  S\;[\![\ \delta_2 \cdots \delta_n\ ]\!]\ (\,\pi\; \cup\; R\;[\![\ \delta_1\ ]\!]\ \pi\ ) & \textbf{if } n > 0\\
  \pi & \textbf{otherwise}
\end{cases}
\end{align*}
%
\Infer specifications may contain recursion in the rules, so one application of the specification may not be sufficient to produce all of the intervals.
The interpretation function $\traceinterpi{\_}{}$ for \textit{inc}lusive \nfer defines the semantics of a specification on a pool by applying $S$ until the inflationary fixed point is reached.
{\needspace{5\baselineskip} 
\begin{lstlisting}
  TraceSemI[:_:] : DeltaList -> Pool -> Pool
  TraceSemI[:delta1 ... deltaN:]pi = UNIONAllI piI:-pi1 = pi && piJ = SpecSem[:delta1 ... deltaN:](piI)
\end{lstlisting}
To maintain consistency with prior work and simplify our presentation, we also overload $\traceinterpi{\_}{}$ to operate on a trace of events ${\tau \in \eventty^*}$ by first converting $\tau$ to the pool $\{ \initf{e} : e \text{ is an element of } \tau \}$ where $\initf{\eta,t,M} = (\eta,t,t,M)$.



\begin{example}
\label{example:intro}
Here, we present an example of an \infer specification with rules useful for our complexity analysis.
Fix $\Ident = \set{\eta_j \mid 0 \le j \le n} \cup \set{d}$ and consider the specification~${D_n = \delta_1 \cdots \delta_n}$ where $\delta_j $ is the rule
\[\scalebox{.9}{$
\coinciderule{\eta_{j+1}}{\eta_{j}}{{\eta_{j}\!}}{\!{m_1,m_2 \mapsto m_1 \shorteq\, m_2\!}}{{\!m_1,m_2 \mapsto\!\{d \mapsto m_1(d)^2\}}}.
$}\]
Here, $m_1$ and $m_2$ denote the maps of the intervals matched by the left and right side of the coincide operator and $d$ represents the only element in their domain.

When applying this specification to the trace~$\tau = (\eta_0, 0, \set{d \mapsto 2})$ we obtain
\[\scalebox{.9}{$
\traceinterpi{D_n}{\tau} = \set{ (\eta_0, 0,0,\set{d \mapsto 2}), (\eta_1, 0,0,\set{d \mapsto 4}), \ldots, (\eta_n, 0,0,\set{d \mapsto 2^{2^n}})}.
$}\]
\end{example}
\begin{remark}
In many of our lower bound proofs, the timestamps of intervals are irrelevant. 
For the sake of readability, we will therefore often disregard the timestamps and denote intervals by $(\eta, y_0, \ldots, y_k)$ where $\set{y_0, \ldots, y_k}$ is the image of the map function of the interval. 
Here, we assume a fixed order of the map domain that will be clear from context.

Also, note that the rules~$\delta_j$ in Example~\ref{example:intro} 
produce an interval~$i'$ labeled by $\eta_{j+1}$ from an interval~$i$ such that
$i$ and $i'$ have the same timestamps and the map value of $i'$ is obtained by squaring the map value of $i$.
Many of the rules we use in our lower bounds proofs have this format.
Again, for the sake of readability, we will not spell out those rules but instead say that the rule produces the interval~$(\eta_{j+1}, y^2)$ from an interval of the form~$(\eta_{j}, y)$.
\end{remark}

\subsection{The Evaluation Problem}
We are interested in the \nfer evaluation problem: Given a specification~$D$, a trace~$\tau$ of events, and a target identifier~$\target$, is there an $\target$-labeled interval in $\traceinterpi{D}{\tau}$?
 Here, we measure the size of a single rule in $D$ by the sum of the length of its map predicate and map update measured in their number of arithmetic and logical operators, with numbers encoded in binary.
The size of an event is the sum of the binary encodings of its timestamps and its map values.
We disregard the identifiers, as their number is bounded by the number of events in the input trace and the number of rules.

\section{Complexity Results for Inclusive nfer}
\label{sec:complexity}
In this section, we determine the complexity of the \infer evaluation problem.
In its most general form it is shown to be undecidable, but we show decidability for three natural fragments.

The undecidability result relies on the recursive nature of \infer, i.e., an $\eta$-labeled interval can be (directly or indirectly) produced from an another $\eta$-labeled interval, and on the fact that the map functions range over the natural numbers, i.e., we have access to an infinite data domain.

\begin{theorem}
\label{thm:full}
The evaluation problem for \infer{} is undecidable.
\end{theorem}

\begin{proof}
We show how to simulate a two-counter Minsky machine~\cite{minsky1967computation} with \infer rules so that the machine terminates iff an interval with a given target identifier can be generated by the rules.

Formally, a two-counter Minsky machine is a sequence
\[
(0:  \instr_0) (1:  \instr_1) \cdots (k-2:  \instr_{k-2})(k-1:  \stopp), 
\]
of pairs~$(\ell: \instr_\ell)$ where $\ell$ is a line number and $\instr_\ell$ for $0 \le \ell < k-1$ is one of
$\inc{i}$, $\dec{i}$, or $\ite{i}{\ell'}$ with $i \in\set{0,1}$ and $\ell'\in \set{0, \cdots, k-1}$. 

A configuration of the machine is a triple~$(\ell, c_0, c_1)$ consisting of a line number~$\ell$ and the contents~$c_i \in \natty$ of counter~$i$.
The semantics is defined as expected with the convention that a decrement of a zero counter has no effect.
The problem of deciding whether the unique run of a given two-counter Minsky machine starting in the initial configuration~$(0,0,0)$ reaches a stopping configuration (i.e., one of the form~$(k-1, c_0, c_1)$) is undecidable~\cite{minsky1967computation}.

This problem is captured with \infer as follows:
We encode a configuration~$(\ell, c_0, c_1)$ by an interval with identifier~$\ell$ and two map values~$c_0, c_1$. 
These intervals use the same timestamps so we drop them from our notation and also write $(\ell, c_0, c_1)$ for the interval encoding that configuration.

For every line number~$0 \le \ell < k-1$ we have one or two rules that are defined as follows (here, we only consider instruction for the first counter, the rules for the second counter are analogous):
\begin{itemize}
	\item $\instr_\ell = \inc{0}$: We have a rule producing the interval~$(\ell+1, c_0+1, c_1)$ from an interval of the form~$(\ell, c_0, c_1)$.
	 \item $\instr_\ell = \dec{0}$: We have two rules, one producing the interval~$(\ell+1, c_0-1, c_1)$ from an interval of the form~$(\ell, c_0, c_1)$ with $c_0 >0$, and one producing the interval~$(\ell+1, c_0, c_1)$ from an interval of the form~$(\ell, c_0, c_1) $ with $c_0 = 0 $.
	 \item $\instr_\ell = \ite{0}{\ell'}$: We have two rules, one producing the interval~$(\ell', c_0, c_1)$ from an interval of the form~$(\ell, c_0, c_1)$ with $c_0 = 0$, and one producing the interval~$(\ell+1, c_0, c_1)$ from an interval of the form~$(\ell, c_0, c_1)$ with $c_0 > 0$.
\end{itemize}
Then, we have an interval labeled by $k-1$ in the fixed point iff the machine reaches a stopping configuration.
\end{proof}

As already discussed, the undecidability relies both on recursion in the rules and on the map functions having an infinite range. 
In the following, we show that restricting one of these two aspects allows us to recover decidability.
In fact, we give tight complexity bounds for all three fragments. 
We continue by introducing some necessary notation to formalize these two restrictions.

First, recall that a map of an interval is a partial function from $\Ident$ to $\natty \cup \boolty$, i.e., it has an infinite range.
We will consider the evaluation problem restricted to intervals with maps that are partial functions from $\Ident$ to $\set{0,1,\ldots,k-1} \cup \boolty$ with a \emph{bound}~$k$ given in binary and all arithmetic operations performed modulo $k$.
We denote the fixed point resulting from these semantics by $\traceinterpi[k]{\_}{}$.

Second, for a rule~$\nferrule{\eta}{\eta_1}{\oplus}{\eta_2}{\Phi}{\Psi}$ we say that $\eta$ appears on the left-hand side and the $\eta_i$ appear on the right-hand side.
An \infer specification~$D \in \rulety^*$ forms a directed graph $G(D)$ over the rules in $D$ such that there is an edge from $\delta$ to $\delta'$ iff there is an identifier~$\eta$ that appears on the left-hand side of $\delta$ and the right-hand side of $\delta'$.
We say that $D$ contains a cycle if $G(D)$ contains one; otherwise $D$ is cycle-free.

\begin{example}
\label{example:evenodd}
Here, we present an example of an \infer specification that captures the natural, even, and odd numbers via intervals.
There are many ways to do so with \nfer rules, but here we choose rules that serve to demonstrate the graph construction described above.
Note that, like in our proofs, we use \coincidekw{}~rules and only have intervals that begin and end at timestamp zero.
As such, we omit timestamps in our notation.
Also, all intervals we consider have a single map value~$v$ storing an integer.

We define the following rules to compute intervals carrying values in the sets of naturals (labeled $N$), evens ($E$), and odds ($O$).
To simplify our presentation, we define a map predicate and two map updates.
The map predicate ${\Phi_{\text{eq}} = m_1,m_2 \mapsto m_1(v) = m_2(v)}$ tests that the map values of the two intervals are equal.
The map update ${\Psi_{\text{zero}} = m_1,m_2 \mapsto\!\{v \mapsto 0\}}$ sets the map value to zero and the map update ${\Psi_{\text{incr}} = m_1,m_2 \mapsto\!\{v \mapsto m_1(v) + 1\}}$ sets the map value to one more than the value of the (equal, assuming $\Phi_{\text{eq}}$ holds) value from the matched intervals.

Now, consider the input trace~$\tau = (I,\set{v \mapsto 0})$ and the rules
\begin{align*}
\scalebox{.9}{$\delta_1 =$} &\ \scalebox{.9}{$\coinciderule{N}{I}{I}{\Phi_{\text{eq}}}{\Psi_{\text{zero}}}$}\\
\scalebox{.9}{$\delta_2 =$} &\ \scalebox{.9}{$\coinciderule{E}{I}{I}{\Phi_{\text{eq}}}{\Psi_{\text{zero}}}$}\\
\scalebox{.9}{$\delta_3 =$} &\ \scalebox{.9}{$\coinciderule{N}{N}{N}{\Phi_{\text{eq}}}{\Psi_{\text{incr}}}$}\\
\scalebox{.9}{$\delta_4 =$} &\ \scalebox{.9}{$\coinciderule{E}{O}{N}{\Phi_{\text{eq}}}{\Psi_{\text{incr}}}$}\\
\scalebox{.9}{$\delta_5 =$} &\ \scalebox{.9}{$\coinciderule{O}{E}{N}{\Phi_{\text{eq}}}{\Psi_{\text{incr}}}$}.
\end{align*}
These rules form the directed graph shown in Figure~\ref{fig:evenodd}.
In the graph, edges are labeled by the identifier witnessing their existence.
For example, there is an edge from $\delta_2$ to $\delta_5$ labeled by $E$ because $E$ apears on the left-hand side of $\delta_2$ and the right-hand side of $\delta_5$.
The graph contains two cycles: one between $\delta_4$ and $\delta_5$ and one with only $\delta_3$, due to its self-loop.

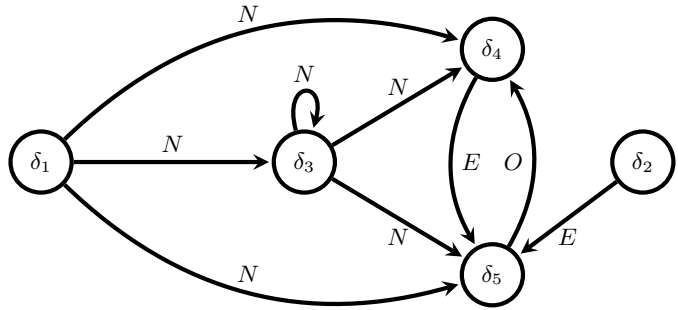
\begin{figure}[h]
  \centering
  \tikzset{>=stealth}
  \begin{tikzpicture}[shorten >=1pt,node distance=2cm,on grid,auto,ultra thick] 
   \node[state] (r1)  {$\delta_1$}; 
   \node[state] (r2)  [right=8cm of r1] {$\delta_2$};
   \node[state] (r3)  [right=3.5cm of r1] {$\delta_3$};
   \node[state] (r4)  [above right=1.5cm and 6cm of r1] {$\delta_4$};
   \node[state] (r5)  [below right=1.5cm and 6cm of r1] {$\delta_5$};
    \path[->] 
    (r1) edge node[above=0] {$N$} (r3)
         edge [bend left=30] node[above=0] {$N$} (r4)
         edge [bend right=30] node[above=0] {$N$} (r5)
    (r2) edge node[below=0] {$E$} (r5)
    (r3) edge [loop above=20] node[above=0] {$N$} ()
         edge node[above=0] {$N$} (r4)
         edge node[below=0] {$N$} (r5)
    (r4) edge [bend right] node {$E$} (r5)
    (r5) edge [bend right] node {$O$} (r4);
  \end{tikzpicture}
  \caption{The directed graph formed by the rules to compute naturals, evens, and odds}
  \label{fig:evenodd}
\end{figure}

Starting with the initial trace~$\tau = (I,\{v \mapsto 0\})$ we obtain the fixed point
\begin{align*}
\traceinterpi{\delta_1 \cdots \delta_5}{(\tau)}{} = &\set{(N,\set{v\mapsto n}) \mid n\in\natty}\, \cup\\
& \set{(E,\set{v\mapsto n}) \mid n\in\natty \text{ is even}}\, \cup\\
& \set{(O,\set{v\mapsto n}) \mid n\in\natty \text{ is odd}} \cup \set{(I, \set{v\mapsto 0})}.
\end{align*}
%
\end{example}

We begin our study of decidable fragments of \infer by considering both the restriction to cycle-free specifications and to finite data at the same time.

\begin{theorem}
\label{thm:inc_free_finite}
The cycle-free \infer evaluation problem with finite data is \pspace-complete.
\end{theorem}

\begin{proof}
We only prove the lower bound here, the upper bound is shown for full \nfer in Theorem~\ref{thm:excl_free_finite}.
We proceed by a reduction from TQBF, the problem of determining whether a formula of quantified propositional logic evaluates to true (see, e.g., \cite{BB} for a detailed definition), which is \pspace-hard.
So, fix such a formula~$\varphi$. 
Let $\pi_j$ for $j \ge 1$ denote the $j$-th prime number.
We assume without loss of generality that $\varphi = 
Q_2 x_2 Q_3 x_3 \cdots Q_{\pi_n} x_{\pi_n} \bigwedge\nolimits_{i = 1}^m ( \ell_{i,1} \vee \ell_{i,2} \vee \ell_{i,3} )
$
where each $Q_{\pi_j}$ is in $\set{\exists,\forall}$, and each $\ell_{i,i'}$ is either $x_{\pi_j}$ or $\neg x_{\pi_j}$ for some $j$.
As we label variables by prime numbers, we can uniquely identify a variable valuation~${V \subseteq \set{x_{\pi_j} \mid 1 \le j \le n}}$ by the number $\prod_{x_{\pi_j} \in V} \pi_j$.
As the map values we will consider only have to encode valuations, and are therefore bounded by $\prod_{j \le n}\pi_j$, we can use the bound~$1 + \prod_{j \le n}\pi_j$ on the map values we consider. 

We present three types of rules:
\begin{enumerate}
	\item Rules to \emph{generate} every possible variable valuation (encoded by an interval whose map contains the number representing the valuation).
	\item A rule to \emph{check} whether a valuation satisfies $\bigwedge_{i = 1}^n ( \ell_{i,1} \vee \ell_{i,2} \vee \ell_{i,3} )$. 
	\item Rules to simulate the quantifier prefix to check whether the full formula evaluates to true. 
\end{enumerate}

Let us explain all steps in detail.
As all intervals in this proof will have the same timestamps, we will drop those to simplify our notation. 
Furthermore, the map of an interval will contain a single integer value. 
For these reasons, we denote intervals by $(\eta, s)$ where $\eta$ is an identifier and $s$ is the map value.

To generate the valuations, we start with the trace containing only a single fixed event that yields the interval~$(G_0,1)$.
Further, for $1 \le j \le n$ we have rules producing the intervals~$(G_j, s \cdot \pi_j)$ and $(G_j, s )$ from an interval of the form~$(G_{j-1},s)$ for some~$s$.
The fixed point reached by applying these rules contains the $2^n$ intervals of the form~$(G_n,s)$ where $s$ encodes a variable valuation. 

In the valuation encoded by some $s$, a variable~$x_{\pi_j}$ evaluates to true if $s \bmod \pi_j = 0$ and evaluates to false if $s \bmod \pi_j \neq 0$. 
Hence, to check whether the valuation encoded by some $s$ satisfies $\bigwedge_{i = 1}^m ( \ell_{i,1} \vee \ell_{i,2} \vee \ell_{i,3} )$
we have a rule that produces the interval~$(C_n, s)$ from an interval of the form~$(G_n, s)$ for some $s$ such that $\bigwedge_{i = 1}^m ( \psi_{i,1} \vee \psi_{i,2} \vee \psi_{i,3} )$ evaluates to true, where $\psi_{i,i'}$ is equal to $s \bmod \pi_{j} = 0$ if $\ell_{i,i'} = x_{\pi_j}$, and where $\psi_{i,i'}$ is equal to $s \bmod \pi_{j} > 0$ if $\ell_{i,i'} = \neg x_{\pi_j}$.

We now simulate the quantifier prefix.
Intuitively, we check whether partial variable valuations cause the formula to hold. 
We do so by the following rules:
If the variable~$x_{\pi_j}$ is existentially quantified, we have a rule producing the interval~$
(C_{j-1}, s)$ from an interval of the form~$(C_{j}, s)$ with $s \bmod \pi_j  > 0$, and a rule producing the interval~$
(C_{j-1}, s/\pi_j)$ from an interval of the form~$(C_{j}, s)$  with $s \bmod \pi_j = 0$.
So, to generate an interval labeled by $C_{j-1}$ at least one interval labeled by $C_{j}$ has to exist, and their maps must be compatible.

Finally, 
if the variable~$x_{\pi_j}$ is universally quantified, we have a rule producing the interval~$(C_{j-1}, s)$ from two intervals of the form~$(C_{j}, s)$ and $(C_{j}, s\cdot \pi_j)$ (which can be done using a \coincidekw-rule).
Thus, to obtain an interval labeled by $C_{j-1}$ both intervals labeled by $C_{j}$ with corresponding map values have to exist.

An induction shows that a partial valuation~$V \subseteq \{x_{\pi_j} \mid 1 \le j \le n'\}$ for some $0 \le n' \le n$ satisfies 
$
Q_{\pi_{n'+1}} x_{\pi_{n'+1}} \cdots Q_{\pi_n} x_{\pi_n} \bigwedge\nolimits_{i = 1}^m ( \ell_{i,1} \vee \ell_{i,2} \vee \ell_{i,3} )
$
iff the interval~$(C_{n'}, \prod_{x_{\pi_j} \in V}\pi_j)$ is generated by applying these rules. 
So, for $n' = 0$ we obtain the correctness of our reduction: The formula~$\varphi$ evaluates to true iff $(C_0, 1)$ is in the fixed point induced by the rules above. 

Furthermore, the rules above are cycle-free, there are linearly many rules in the number~$n$ of variables and each rule is of polynomial size in the size of $\varphi$.
Finally, as $\pi_j \le j (\ln j + \ln\ln j)$ for all $j \ge 6$~\cite{primes}, all numbers appearing in the maps of the intervals are bounded by 
\[\prod\nolimits_{j = 1}^n \pi_j \le c \cdot \prod\nolimits_{j = 1}^n j (\ln j + \ln\ln j) \le c \cdot (n (\ln n + \ln\ln n))^n\]
whose binary representation is polynomial in the size of $\varphi$.
Here, $c$ is some constant that is independent of $n$.
\end{proof}

Now, we turn our attention to the remaining two fragments obtained by considering finite-data with cycles and cycle-free specifications with infinite data. 
In both cases, we again prove tight complexity bounds.
For both upper bounds, we rely on algorithms searching for witnesses for the existence of an interval in the fixed point.
As these arguments are used in multiple proofs, we introduce them first in a general format.
So, fix some specification~$D$ and some trace $\tau$ of events.
If $i$ is an interval in $\traceinterpi{D}{\tau}$, then either there is an event $e$ in $\tau$ such that $\initf{e} = i$ (we say that $i$ is initial in this case) or there are intervals~$i_1, i_2$ in $\traceinterpi{D}{\tau}$ and a rule~$\delta \in D$ such that $i$ is obtained by applying $\delta$ to $i_1$ and $i_2$.
So, for every interval~$i_0$ in $\traceinterpi{D}{\tau}$ there is a binary (witness) tree whose nodes are labeled by intervals in $\traceinterpi{D}{\tau}$, whose root is labeled by $i_0$, whose leaves are labeled by initial intervals, and where the children of a node labeled by $i$ are labeled by $i_1$ and $i_2$ such that there is a rule~$\delta$ so that $i$ is obtained by applying $\delta$ to $i_1$ and $i_2$.
Note that the tree might contain multiple occurrences of the same interval. But, we can assume without loss of generality that each path in the tree does not contain a repetition of an interval (if it does we can just remove the part of the tree between the repetitions).
Hence, the height of the tree is bounded by the number of intervals, which might be infinite in the case of infinite data.
Furthermore, if $D$ is cycle-free then the height of a witness tree is also bounded by the number of rules in $D$.
Note that the same arguments also apply to $\traceinterpi[k]{D}{\tau}$ in case we deal with finite data.

\begin{proposition}
\label{prop:witnesstrees}
An interval is in $\traceinterpi{D}{\tau}$ ($\traceinterpi[k]{D}{\tau}$) iff it has a witness tree.	
\end{proposition}

We continue by settling the case of specifications with cycles, but restricted to finite data.

\begin{theorem}
\label{thm:inc_finite}
The \infer evaluation problem with finite data is \exptime-complete.
\end{theorem}

\begin{proof}
We first prove the lower bound by reducing from the word problem for alternating polynomial space Turing machines (see, e.g.,~\cite{BB} for detailed definitions). 
As $\exptime = \apspace$, this yields the desired lower bound.

Recall that an alternating Turing machine has both existential states and universal states and a run of such a machine is a tree labeled by configurations.
An accepting run is finite, labeled by the initial configuration in the root, every vertex labeled by an existential configuration (i.e., one with an existential state) has one child labeled by a successor configuration, every vertex labeled by a universal configuration has for all successor configurations a child labeled by it, and all leaves are labeled by accepting configurations.
We refer to paths through such a (tree-like) run as run-branches. 

So, fix an alternating polynomial space Turing machine~$\tm$, i.e., there is some polynomial~$p$ such that $\tm$ uses at most space~$p(|w|)$ when started on input~$w$.
Let us also fix some input~$w$ for $\tm$. 
We construct an instance of \infer that simulates a run of $\tm$ on $w$.
To simplify our construction, we make some assumptions (all without loss of generality):
\begin{itemize}
	\item The set~$Q$ of states of $\tm$ is of the form~$\set{1,2,\ldots, |Q|}$ and $1$ is the initial state.
	\item The tape alphabet~$\Gamma$ of $\tm$ is equal to $\set{0,1, \ldots, 9}$ and $0$ is the blank symbol.
	\item Every run of $\tm$ has only finite branches, i.e., $\tm$ terminates on every input. To this end, we assume the existence of a set of terminal states, which is split into accepting and rejecting ones.
	\item Every nonterminal configuration (one with a nonterminal state) has exactly two successor configurations. Such states are either existential or universal.
\end{itemize}
So, a configuration of $\tm$ is of the form~$\cleft q \cright$ with $q \in Q$ and $\cleft, \cright \in \Gamma^*$ such that $|\cleft|+|\cright|=p(|w|)$, with the convention that the head is on the first letter of $\cright$.

For $c \in \Gamma^*$, let $c^R$ denote the reverse of $c$.
Due to our assumption on $\Gamma$ we can treat $\cleft$ and $\cright^R$ as natural numbers encoded in base ten.
We uniquely identify a configuration~$\cleft q \cright$ by the triple~$(\cleft, q, \cright^R)$ of natural numbers. 
The initial configuration of $\tm$ on $w$ is encoded by the triple $(0,1,w^R)$ representing that the tape to the left of the head has only blanks, the machine is in the initial state~$1$, and $w$ is to the right of the head with the remaining cells of the tape being blank.

This encoding allows us to read the tape cell the head is currently pointing to, update the tape cell the head is pointing to, and move the head by simple arithmetic operations.
For example, whether the head points to a cell containing a $3$ is captured by $\cright^R \bmod 10$ being $3$, and writing a $7$ to the cell pointed to by the head is captured by adding $- (\cright^R \bmod 10) + 7 $ to $\cright^R$.
Finally, moving the head to, say, the right, is captured by multiplying $\cleft$ by $10$ and then adding $\cright^R \bmod 10$ to it, and then dividing $\cright^R$ by $10$ (which is done without remainder and therefore removes the last digit of $\cright^R$).
In the following, we use intervals of the form~$(A, \cleft, q, \cright^R)$ to encode the configuration $\cleft q \cright$ of $\tm$. 
Here, $A$ is some identifier and we disregard timestamps, as all intervals have the same start and end. 
Hence, $\cleft$, $ q$, and $\cright^R$ are three map values of the interval.

We now describe the rules simulating $\tm$ on $w$. 
We start with some fixed event that yields the interval~$(G, 0,1,w^R)$ encoding the initial configuration.
As described above, the computation of a successor configuration can be implemented using arithmetic operations. 
Thus, given the interval encoding the initial configuration, one can write rules (one for each transition of $\tm$) that generate the set of all configurations, encoded as intervals of the form~$(G, \cleft, q, \cright^R)$.
Furthermore, one can write a rule that produces the interval~$(A, \cleft, q, \cright^R)$ from every interval~$(G, \cleft, q, \cright^R)$ with an accepting $q$.

Now, we describe rules to compute the set of accepting configurations, i.e., the smallest set~$A$ of configurations that contains all those with an accepting terminal state, all existential ones that have a successor in $A$, and all universal ones that have both successors in $A$.
For every transition~$t$ from an existential state~$q$, there is a rule to produce the interval~$(A, \cleft, q, \cright^R)$ if the intervals~$(G, \cleft, q, \cright^R)$ and $(A, \cleft', q', {\cright^R}')$ already exist, where $(\cleft', q', {\cright^R}')$ encodes the configuration obtained by applying the transition~$t$ to the configuration encoded by $(\cleft, q, \cright^R)$. 
Thus, to declare an existential configuration as accepting at least one of its successor configurations has to be already declared as accepting. 

Now, let us consider universal configurations.
Due to our assumption, for every pair of a state and a tape symbol, there are exactly two transitions~$t_1$ and $t_2$ that are applicable.
There are two rules for this situation.
The first one produces the interval~$(B, \cleft, q, \cright^R)$ if the intervals~$(G, \cleft, q, \cright^R)$ and $(A, \cleft', q', {\cright^R}')$ already exist, where $(\cleft', q', {\cright^R}')$ encodes the configuration obtained by applying the transition~$t_1$ to the configuration encoded by $(\cleft, q, \cright^R)$. 
The second one produces the interval~$(A, \cleft, q, \cright^R)$ if the intervals~$(B, \cleft, q, \cright^R)$ and $(A, \cleft', q', {\cright^R}')$ already exist, where $(\cleft', q', {\cright^R}')$ encodes the configuration obtained by applying the transition~$t_2$ to the configuration encoded by $(\cleft, q, \cright^R)$. 
Thus, to declare a universal configuration as accepting both of its successor configurations have to be already declared as accepting. 

Finally, there is a rule producing an interval with identifier~$\target$ from the interval~$(A,0,1,w^R)$, indicating that the initial configuration is accepting. 
Thus, the fixed point contains an interval labeled by $\target$ iff $\tm$ accepts~$w$. 

It remains to show that the specification has the required properties.
It is of polynomial size and each rule has polynomial size (both measured in $|\tm| + |w|$). 
Further, all numbers used in the intervals are bounded by $\max\set{|Q|, 10^{p(|w|)}}$, whose binary representation is bounded polynomially in $|\tm| + |w|$.

Now, we prove the upper bound. 
We are given a specification~$D$, an input trace~$\tau$ of events, a $k \in \natty$ (given in binary), and a target label~$\target$ and have to determine whether the fixed point~$\traceinterpi[k]{D}{\tau}$ contains an interval labeled by $\target$.
We describe an alternating polynomial space Turing machine solving this problem by searching for a witness tree. $\apspace = \exptime$ yields the result.

To this end, we rely on the following properties.
\begin{enumerate}
	\item Every interval in $\traceinterpi[k]{D}{\tau}$ can be stored in polynomial space, as every value in its map can be stored using $\log k$ bits, and there are only linearly many such values (measured in $|D| + |\tau|$).
	\item There are only exponentially many intervals in $\traceinterpi[k]{D}{\tau}$, e.g., 
	\[b(D, \tau, k) = \iota \cdot t^2 \cdot k^{|D| + |\tau|} \le  \iota \cdot |\tau|^2 \cdot 2^{(\log k) (|D| + |\tau|)}\] is a crude upper bound. Here, $\iota$ is the number of identifiers appearing in $D$ and $\tau$ and $t$ is the number of timestamps in $\tau$ (recall that \infer does not create new timestamps).
	\item Given three intervals~$i, i_1, i_2$ and a rule~$\delta \in D$ one can determine in polynomial space whether $i$ is obtained by applying $\delta$ to $i_1$ and $i_2$.
\end{enumerate}
\begin{algorithm}[t]
\caption{Algorithm checking the existence of a witness tree}
\label{algo_cyclesfin}
\begin{algorithmic}[1]
\REQUIRE{Specification~$D$, trace~$\tau$, bound~$k$, target identifier~$\target$}
\STATE $n := 0$	
\STATE \textbf{nondeterministically guess} interval~$i$ labeled by $\target$
\WHILE{$n < b(D, \tau, k)$ \AND $i$ is not initial}
\STATE $n := n + 1$
\STATE \textbf{nondeterministically guess} intervals $i_1, i_2$ and $\delta \in D$ such that $i$ is obtained by applying $\delta$ to $i_1$ and $i_2$
\STATE \textbf{universally pick} $i := i_j$ for $j \in {1,2}$
\ENDWHILE
\STATE \textbf{if} $i$ is initial \textbf{then return} accept
\STATE \textbf{else} \textbf{return} reject
\end{algorithmic}
\end{algorithm}

Using alternation, Algorithm~\ref{algo_cyclesfin} determines whether a witness tree exists whose root is labeled by $\target$ and whose height is bounded by $b(D, \tau, k)$.
Due to Proposition~\ref{prop:witnesstrees}, this is equivalent to an interval labeled by $\target$ being in $\traceinterpi[k]{D}{\tau}$.
Due to the above properties, one can easily implement the algorithm on an alternating polynomial space Turing machine, yielding the desired upper bound.
\end{proof}

Finally, we consider the last fragment: cycle-free specifications with infinite data. 
A crucial aspect here is that cycle-free specifications imply an upper bound on the map values of intervals in the fixed point, as each interval in the fixed point can be generated by applying each rule at most once.
For the lower bound, we generate \emph{large} numbers using a set of cycle-free rules (cf.\ Example~\ref{example:intro}) and encode configurations using these numbers as before. 

\begin{theorem}
\label{thm:inc_free_infinite}
The cycle-free \infer evaluation problem with infinite data is \nexptime-complete.
\end{theorem}

\begin{proof}
We begin with the lower bound and reduce from the word problem for nondeterministic exponential-time Turing machines, which is \nexptime-hard.
Thus, fix such a machine~$\tm$, i.e., there is some polynomial~$p$ such that $\tm$ uses at most time~$2^{p(|w|)}$ when started on input~$w$.
Let us also fix some input~$w$ for $\tm$ and define $n = p(|w|)$.
We construct an instance of \infer that simulates the run of $\tm$ on $w$.
To simplify our construction, we make some assumptions (all without loss of generality):
\begin{itemize}
	\item The set~$Q$ of states of $\tm$ is of the form~$\set{1,2,\ldots, |Q|}$ and $1$ is the initial state.
	\item The tape alphabet~$\Gamma$ of $\tm$ is equal to $\set{0,1, \ldots, 9}$ and $0$ is the blank symbol.
	\item Accepting states are equipped with a self-loop. This allows us to restrict our attention to runs of length~$2^n$.
\end{itemize}
A configuration of $\tm$ is of the form~$\cleft q \cright$ with $q \in Q$ and $\cleft, \cright \in \Gamma^*$ such that $|\cleft|+|\cright|\le 2^n$, with the convention that the head of $\tm$ is on the first letter of $\cright$.
So, all the numbers encoding configurations are bounded by $b = 2^{2^{n+2}} \ge 10^{2^{n}}$.

First we describe how to generate all natural numbers up to $b$ using cycle-free rules (as before, we drop the timestamps and only write the identifier and map values).
This construction is based on the fact that squaring two $k$ times yields $2^{2^k}$ (cf.~Example~\ref{example:intro}):
We start with an initial interval~$(N_0, 1)$ (coming from the input trace) and have, for every $1 \le j \le n+2$, a rule producing the interval~$(N_{j},s^2 + s')$ from the intervals $(N_{j-1}, s)$ and $(N_{j-1}, s')$.
An induction shows that applying these rules until a fixed point is reached generates all intervals~$(N_{n+2},s)$ with $1 \le s \le b$.
Thus, we can also write a rule generating all intervals of the form~$(C, \cleft, q, \cright^R)$ with $\cleft + \cright \le b$.
Thus, each $C$-labeled interval encodes a configuration of $\tm$.

Now, we piece together a run out of configurations as follows:
We have a rule producing the interval~$(C_{0}, \cleft, q, \cright^R, \cleft, q, \cright^R)$ from an interval of the form~$(C, \cleft, q, \cright^R)$.
The resulting interval represents a run infix of length~$1 = 2^0$.

For all $0 < j \le n$ we have a rule producing the interval
\[
(C_j, \cleft_0, q_0, \cright^R_0, \cleft_3, q_3, \cright^R_3) \] 
from two intervals of the form
\[(C_{j-1}, \cleft_0, q, \cright^R_0, \cleft_1, q_1, \cright^R_1)\]
and
\[ (C_{j-1}, \cleft_2, q_2, \cright^R_2, \cleft_3, q_3, \cright^R_3)\] 
such that $\cleft_2 q_2 \cright_2$ encodes a successor configuration of the configuration encoded by $\cleft_1 q_1 \cright_1$.
These rules allow to combine two run infixes of length~$2^{j-1}$ into a one of length~$2^j$, provided the second one starts with a successor configuration of the first one.
As we can double the length of the run infix with each rule application, a linear number of rules suffices to construct a full run of exponential length.

Thus, we have a rule producing the interval~$(\target)$ from an interval of the form 
\[(C_0, \cleft_0, q_0, \cright^R_0, \cleft_1, q_1, \cright^R_1)\] such that $(\cleft_0, q_0, \cright^R_0)$ encodes the initial configuration of $\tm$ on $w$ and $q_1$ is an accepting state.
An induction shows that an interval labeled with $\target$ is generated iff $\tm$ has an accepting run on $w$.

For the matching upper bound, we show that a witness tree for an interval can be found in nondeterministic exponential time.
This is the based on the following properties for a given specification~$D$ and a given trace~$\tau$ of events.
As $D$ is cycle-free, every interval in $\traceinterpi{D}{\tau}$ has a witness tree whose height is bounded by $|D|$.
Now, if an interval~$i$ is obtained from intervals~$i_1,i_2$ by the application of some rule~$\delta \in D$, the maximal map value in $i$ is only polynomial larger than the maximal map value in $i_1$ and $i_2$, with the polynomial only depending on $\delta$.
This is because map updates are limited to simple arithmetic operations including multiplication but excluding exponentiation.

Hence, there is some polynomial~$p$ (dependent on $D$ and $\tau$) such that all map values of intervals in $\traceinterpi{D}{\tau}$ are bounded by $2^{2^{p(|D| + |\tau|)}}$.
	Thus, every interval in $\traceinterpi{D}{\tau}$ can be stored (with map values encoded in binary) in exponential space (measured in the size of $D$ and $\tau$).

Furthermore, as the height of a witness tree (which is a binary tree) can be bounded by $|D|$, its size can be bounded by $2^{|D|}$.
Combining both bounds shows that each interval in $\traceinterpi{D}{\tau}$ has a witness tree that can be encoded using an exponential number of bits.
Hence, given $D$, $\tau$, and a target label~$\target$ we can guess a tree and verify that it is indeed a witness tree for some $\target$-labeled interval in exponential time. 
\end{proof}

\section{The Full nfer Language}
\label{sec:exclusive}
This section introduces the second type of \nfer rules, called \emph{exclusive rules}, that test for the existence of one interval and the absence of another interval matching constraints.
These rules were introduced in~\cite{kauffman2017inferring} and they, together with inclusive rules, complete the \nfer language.
We define the syntax of these rules, including mathematical functions to simplify the presentation, as follows:
\[\nferrule{\eta}{\eta_1}{\textbf{unless } \ominus}{\eta_2}{\Phi}{\Psi}\]

\noindent
where $\eta, \eta_1, \eta_2 \in \Ident$ are identifiers,
$\ominus \in \{\afterkw{},\followkw{},\containkw{}\}$ is a \emph{clock predicate} on two intervals (one for each of $\eta_1\text{ and }\eta_2$), and $\Phi$ and $\Psi$ are the same as in inclusive rules.
We say that an exclusive rule \emph{includes} $\eta_1$ and \emph{excludes} $\eta_2$.

\subsection{Semantics}

Exclusive rules share many features with inclusive rules but they require additions to the \infer{} semantics that were omitted in Section~\ref{sec:language} for brevity.
Notably, these changes to the semantics produce equivalent results when evaluating inclusive rules.
The following definition gives semantics to exclusive rules:
%
\begin{lstlisting}
       RuleSem[:eta <- eta1 unless ominus eta2 wherekw Phi mapkw Psi:]pi = 
            { i0 isin Interval : i1 isin pi :- Idof(i0) = eta && Idof(i1) = eta1 &&
                       Startof(i0) = Startof(i1) &&  Endof(i0) = Endof(i1) &&
                       Mapof(i0) = Psi(Mapof(i1),emptymap) && 
                       not (Exists i2 isin pi :- i2 != i1 && Idof(i2) = eta2  && 
                           ominus(i1,i2) && Phi(Mapof(i1),Mapof(i2)) ) }
\end{lstlisting}
\noindent
Like with inclusive rules, exclusive rules match intervals in the input pool~$\pi$ to produce a pool of new intervals.
The difference is that exclusive rules produce new intervals where one existing interval in $\pi$ matches the identifier~$\eta_1$ and no intervals exist that match the identifier~$\eta_2$ such that the clock predicate~$\ominus$ and the map predicate~$\Phi$ hold for the $\eta_1$-labeled and the $\eta_2$-labeled interval. 

\begin{remark}
This definition is more general than the exclusive rule semantics in~\cite{kauffman2017inferring} where intervals with excluded identifiers were limited to those with end timestamps strictly less than the extant interval (\isem{Endof(i2) < Endof(i1)}).
Here, we instead enforce only that an interval does not exclude itself (\isem{i2 != i1}).
The restriction on end timestamps was originally introduced to facilitate completeness in an online monitoring algorithm but is more restrictive than necessary for the offline semantics we present here.
\end{remark}

The three possibilities referenced by $\ominus$ are shown in Figure~\ref{fig:exclusiveops} and formally defined in Table~\ref{table:exclusiveops}.
These clock predicates relate two intervals using familiar \ac{atl} temporal operators while the timestamps of the produced interval are copied from the included interval rather than being defined by the clock predicate.
In the figure, the excluded interval labeled $C$ is shown as a rectangle with a dotted outline and the produced interval labeled $A$ is always the same as the included interval labeled $B$.
For example, given intervals $i,i_{1},i_{2}$ where ${\Idof{i} = A}$, ${\Idof{i_{1}} = B}$ and ${\Idof{i_{2}} = C}$, ${A \leftarrow \follow{B}{C}}$ holds when ${\Endof{i_{2}} = \Startof{i_{1}}}$, ${\Startof{i} = \Startof{i_1}}$, and ${\Endof{i} = \Endof{i_1}}$.

\begin{figure}[]
\centering
\begin{tabular}{c|c|c}
\allen{unless after} & \allen{unless follow} & \allen{unless contain}\\
\includegraphics[height=1cm]{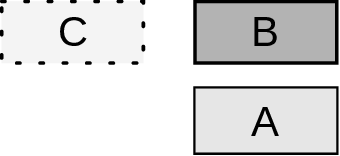} & \includegraphics[height=1cm]{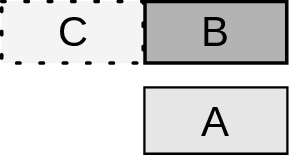} & \includegraphics[height=1.4cm]{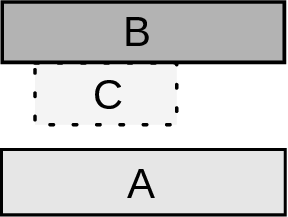}
\end{tabular}
\caption{\nfer clock predicates for exclusive rules}
\label{fig:exclusiveops}
\end{figure}

\begin{table}[]
  \caption{Formal definition of \nfer clock predicates for exclusive rules}
  \label{table:exclusiveops}
  \centering
  \begin{tabular}{llll}
  \toprule
  Syntax & \phantom{} & Definition of $\ominus(i_1,i_2)$ & \phantom{mmmm}\\
  \midrule

  \afterkw{} & \phantom{} & $\Startof{i_1} > \Endof{i_2} $ & \phantom{mmmm}\\
 \rowcolor{gray!30} \followkw{} & \phantom{} & $\Startof{i_1} = \Endof{i_2} $ & \phantom{mmmm}\\
  \containkw{} & \phantom{} & $\Startof{i_2} \geq \Startof{i_1}  \wedge \Endof{i_2} \leq \Endof{i_1}$ & \phantom{mmmm}\\
  \end{tabular}
\end{table}

To ensure that evaluating an \nfer specification yields a unique fixed point, exclusive rules may not appear in cycles because the intervals they produce depend on the persistent non-existence of other intervals.
When cycles exist in an \infer specification, rules are evaluated multiple times and each evaluation may add intervals.
Exclusive rules may have non-deterministic behavior in a cycle because the intervals they exclude may be produced either before or after the exclusive rule is evaluated.
Note that it is possible to evaluate specifications with both cycles \emph{among inclusive rules} and exclusive rules (not appearing in the cycles) with deterministic behavior and in Section~\ref{sec:extended} we introduce an extended semantics that supports this case.
Here, we first consider the case of cycle-free specifications with exclusive rules.

The order in which rules are evaluated may also affect the result of applying exclusive rules for this reason, which motivates a generalization of the $\traceinterpi{\_}{}$ (resp. $\traceinterpi[k]{\_}{}$) function.
\begin{align*}
&T\!{\text{\tiny{full}}}\;[\![\_]\!]\ :\ \rulety^* \rightarrow \poolty \rightarrow \poolty\\
&T\!{\text{\tiny{full}}}\;[\![\ \delta_1\, \cdots \delta_n\ ]\!]\ \pi = \begin{cases}
  S\;[\![\ \textit{topsort}(\delta_1\, \cdots \delta_n)\ ]\!]\ (\,\pi\,) & \textbf{if } \exists\ \textit{topsort}(\delta_1\, \cdots \delta_n)\\
  T\!{\text{\tiny{inc}}}\;[\![\ \delta_1\, \cdots \delta_n\ ]\!]\ (\,\pi\,) & \textbf{otherwise}\\
  \end{cases}
\end{align*}
where \isem{topsort} is a topological sort of the directed graph $G(D)$ described in Section~\ref{sec:complexity} and $\traceinterpi{\_}{}$ is the interpretation function defined in Section~\ref{paragraph:semantics}.
A topological sort, which can be computed in linear time~\cite{kahn1962topological}, only exists in a cycle-free specification.
In that case, \isem{topsort} orders the rules such that the fixed-point computation of $\traceinterpi{\_}{}$ can be short-circuited, since one application of $\seminterp{\_}$ is sufficient to produce the final pool.
The results of $\traceinterpf{\_}{}$ are independent of the topological sort, as any such ordering will guarantee that all intervals matched by a rule exist before it is applied using $\interp{\_}$.


%
\begin{example}\label{example:complement}
From the trace~$(N, \set{v \mapsto 0})\cdots (N,  \set{v \mapsto 100})$, the rule
	\[\scalebox{.9}{$
\coinciderule{E}{N}{{N\!}}{\!{m_1,m_2 \mapsto m_1(v) \bmod 2 = 0\!}}{{\!m_1,m_2 \mapsto\!\{v \mapsto m_1(v)\}}}
$}\]
produces the set~$\set{(E, \set{v\mapsto n}) \mid n \le 100 \text{ is even}}$. Here, $m_i$ for $i \in \set{1,2}$ is the map of the interval matched with the $i$-th $N$ on the right-hand side of the rule.
Then, the (exclusive) rule
\[\scalebox{.9}{$
\nferrule{O}{N}{\unlesskw\ \containkw}{E}{\!{m_1,m_2 \mapsto m_1(v) = m_2(v)\!}}{m_1\mapsto \set{v \mapsto m_1(v)}}
$}\]
yields the set~$\set{(N, \set{v \mapsto n}) \mid n \le 100 \text{ is odd}}$. Here, $m_1$ is the map of the interval matched with $N$ while $m_2$ is the interval matched with $E$.
So, the exclusive rule intuitively describes the complement of the $E$-labeled intervals with respect to the $N$-labeled intervals (which are all identifiable by numbers~$0 \le n \le 100$).
\end{example}

\section{Complexity Results for Exclusive nfer}

In the following, we study the complexity of the cycle-free \nfer evaluation problem with finite and infinite data, starting with the former.

\begin{theorem}
\label{thm:excl_free_finite}
The cycle-free \nfer evaluation problem with finite data is \pspace-complete.
\end{theorem}

\begin{proof}
The lower bound already holds for the special case of \infer (see Theorem~\ref{thm:inc_free_finite}), so we only need to prove the upper bound. 
To this end, we show how to witness in alternating polynomial time that a given interval is in the fixed point, which yields the desired bound due to $\aptime = \pspace$.
Note that we cannot just search for a witness tree as for \infer, as we also have to handle exclusive rules.

Intuitively, an exclusive rule requires the existence of one interval in the fixed point and the non-existence of other intervals in the fixed point. 
We have seen how to capture existence of an interval via the existence of a witness tree. 
Hence, we can capture the non-existence of an interval via the non-existence of a witness tree.
As we construct an alternating algorithm, we use duality to capture the non-existence of a witness tree and switch between an existential and a universal mode every time the non-existence of an interval is to be checked.

\begin{algorithm}
\caption{Solving the cycle-free \nfer evaluation problem with finite data}
\label{algo_cyclesfinexcl}
\small
\begin{algorithmic}[1]
\REQUIRE{Specification~$D$, trace~$\tau$, target identifier~$\target$}
\STATE $n := 0$, $f := 0$
\STATE \textbf{nondeterministically guess} interval~$i$ labeled by $\target$
\WHILE{$n < |D|$ \AND $i$ is not initial}
\STATE $n := n + 1$
\IF{$f = 0$}
\STATE \textbf{nondeterministically guess} rule $\delta \in D$ 
\IF{$\delta$ is inclusive}
\STATE \textbf{nondeterministically guess} intervals $i_1, i_2$ such that $i$ is obtained by applying $\delta$ to $i_1$ and $i_2$
\STATE \textbf{universally pick} $i := i_j$ for $j \in {1,2}$
\ELSE[$\delta$ is exclusive]
\STATE \textbf{nondeterministically guess} interval $i_1$ such that $i$ is obtained by applying $\delta$ to $i_1$
\STATE \textbf{universally pick} an interval $i_2 \in \{ i' : \delta \text{ includes } i_1 \text{ and excludes } i' \} \cup \{\bot\}$
\IF{$i_2 = \bot$}
\STATE $i := i_1$
\ELSE
\STATE \textbf{universally pick} $i := i_j$ for $j \in \set{1,2}$
\IF{$i = i_2$}
\STATE $f := 1 - f$
\ENDIF
\ENDIF
\ENDIF 
\ELSE[$f = 1$]
\STATE \textbf{universally pick} rule $\delta \in D$ 
\IF{$\delta$ is inclusive}
\STATE \textbf{universally pick} intervals $i_1, i_2$ such that $i$ is obtained by applying $\delta$ to $i_1$ and $i_2$
\STATE \textbf{nondeterministically guess} $i := i_j$ for some $j \in {1,2}$
\ELSE[$\delta$ is exclusive]
\STATE \textbf{universally pick} interval $i_1$ such that $i$ is obtained by applying $\delta$ to $i_1$
\STATE \textbf{nondeterministically guess} an interval $i_2 \in \{ i' : \delta \text{ includes } i_1 \text{ and }$ $\text{excludes } i' \} \cup \{\bot\}$
\IF{$i_2 = \bot$}
\STATE $i := i_1$
\ELSE
\STATE \textbf{nondeterministically guess} $i := i_j$ for $j \in \set{1,2}$
\IF{$i = i_2$}
\STATE $f := 1 - f$
\ENDIF
\ENDIF
\ENDIF
\ENDIF
\ENDWHILE
\IF{$i$ is initial \AND $f = 0$}
\RETURN{accept}
\ELSE
\RETURN{reject}
\ENDIF
\end{algorithmic}
\end{algorithm}

Algorithm~\ref{algo_cyclesfinexcl} keeps track of a single interval and applies rules in a backwards fashion. 
Using alternation, it guesses and verifies a tree structure witnessing the (non-)existence of intervals in the fixed point.
To simulate exclusive rules, it uses a Boolean flag~$f$ to keep track of the parity of the number of exclusive rules that have been simulated, initialized with zero.
If $f$ is zero, then a rule~$\delta$ is guessed nondeterministically.
If this rule is inclusive, two intervals~$i_1$ and $i_2$ are guessed nondeterministically such that the current interval~$i$ is obtained from $i_1$ and $i_2$ by applying $\delta$.
Then, the current interval is updated by universally picking $i := i_1$ or $i := i_2$, so that both choices are checked. 
This case is similar to Algorithm~\ref{algo_cyclesfin}.

On the other hand, if the rule is exclusive, then a single interval~$i_1$ is guessed nondeterministically and another interval~$i_2$ is picked universally so that $\delta$ includes $i_1$, excludes $i_2$, and $i$ is the result of applying $\delta$ to $i_1$.
Now, the current interval is updated by universally picking $i := i_1$ or $i := i_2$, so that both choices are checked.
In the second case, the flag is toggled to signify that another exclusive rule is simulated.

In the case where $f$ is equal to one, the approach is just dual, i.e., we switch existential and universal choices.
As the input specification is cycle-free, we need to simulate at most $|D|$ applications of a rule.
Finally, acceptance depends on the value of the flag, i.e., while the flag is zero the last interval has to be initial (i.e., in the input trace) while it has to be non-initial if the flag is one.  

The algorithm runs in alternating polynomial time as each run simulates at most $|D|$ rule applications and each application can be implemented in deterministic polynomial time due to the encodings of the map values and time stamps being bounded by $|D|+|\tau|+\log(k)$, where $k$ is the bound on the map values.
\end{proof}

Finally, we consider the case of infinite data.
Here, the bound we obtain is \aexptimepoly, the class of problems decided by alternating exponential-time Turing machines with a polynomial number of alternations between existential and universal states.

\begin{theorem}
\label{thm:excl_free_infinite}
The cycle-free \nfer evaluation problem with infinite data is \aexptimepoly-complete.
\end{theorem}

\begin{proof}
We begin with the lower bound, showing a reduction from the word problem for alternating exponential-time Turing machines with a polynomial number of alternations.
To this end, fix such a machine~$\tm$, i.e., there are two polynomials $p,p'$ such that $\tm$ uses at most time~$2^{p(|w|)}$ and at most $p'(|w|)$ alternations when started on input~$w$.
We fix an input~$w$ and define $n = p(|w|)$ and $m = p'(|w|)$.
We construct a cycle-free \nfer instance with infinite data that simulates a run of $\tm$ on $w$. 

Recall that in the proof of Theorem~\ref{thm:inc_free_infinite}, we have shown that a \emph{nondeterministic} exponential-time Turing machine can be simulated by a cycle-free \infer instance (assuming infinite data).
Here, we show that exclusive rules can simulate alternation.

As before, we make some simplifying assumptions on $\tm$ (all without loss of generality):
\begin{itemize}
	\item The set~$Q$ of states of $\tm$ is of the form~$\set{1,2, \ldots, |Q|}$ and $1$ is the initial state.
	\item The tape alphabet~$\Gamma$ of $\tm$ is equal to $\set{0,1,\ldots, 9}$ and $0$ is the blank symbol.
	\item All accepting states are existential and every branch of a run of $\tm$ on $w$ has exactly $m$ alternations.
\end{itemize}

As in the proof of Theorem~\ref{thm:inc_free_infinite}, we encode a configuration of $\tm$ as $\cleft q \cright^R$ where $q$ is the state and $\cleft, \cright^R \in \Gamma^*$ are treated as natural numbers bounded by $2^{2^{n+2}}$ encoding the tape to the left of the reading head ($\cleft$) and the reverse of the tape to the right of the reading head ($\cright^R$).

As shown in the proof of Theorem~\ref{thm:inc_free_infinite}, we can construct a set~$D$ of cycle-free inclusive rules and a trace~$\tau$ such that the fixed point~$\traceinterpi{D}{\tau}{}$ encodes all those pairs~$(c,c')$ of configurations of a nondeterministic Turing machine such that $c'$ is reachable from $c$ via a run of length at most~$2^n$.
As usual, here, all intervals have the same time stamps, which is the reason we drop them from the notation.
In this proof, we work with tree-like runs, as $\tm$ is an alternating Turing machine.
In the following, we need to distinguish between branches of such runs using only existential and only universal states.
To this end, we use the identifiers~$\exists$ and $\forall$. We adapt the rules from the proof of Theorem~\ref{thm:inc_free_finite} such that 
\begin{itemize}
	\item the interval~$(\exists, \cleft_0, q_0, \cright_0^R, \cleft_1, q_1, \cright_1^R)$ is in the fixed point iff the configuration~$\cleft_1 q_1 \cright_1$ is reachable from the configuration~$\cleft_0 q_0 \cright_0$ via a branch of length at most~$2^n$ such that all states (including $q_0$ and $q_1$) are existential, and
	\item the interval~$(\forall, \cleft_0, q_0, \cright_0^R, \cleft_1, q_1, \cright_1^R)$ is in the fixed point iff the configuration~$\cleft_1 q_1 \cright_1$ is reachable from the configuration~$\cleft_0 q_0 \cright_0$ via a branch of length at most~$2^n$ such that all states (including $q_0$ and $q_1$) are universal.
\end{itemize}

Recall that $m$ is the number of alternations during every branch of a run on our fixed input~$w$ and that all accepting states are existential. 
In the following, we show how to inductively \emph{compute} the set of configurations from which an accepting subrun of length at most $2^n$ and with $m' \le m$ alternations starts, using cycle-free rules only.

For $m' = 0$, this is the set of existential configurations that can reach an accepting one via existential configurations only.
Hence, we have a rule that turns an interval of the form
\[
(\exists, \cleft_0, q_0, \cright_0^R, \cleft_1, q_1, \cright_1^R)
\]
with accepting $q_1$ into the interval~$
(L_0, \cleft_0, q_0, \cright_0^R)$.

Now, using an exclusive rule we can obtain the intervals of the form
$
(\overline{L_0}, \cleft_0, q_0, \cright_0^R)
$
such that
$
(L_0, \cleft_0, q_0, \cright_0^R)$
has not been generated so far, i.e., we generate the complement (cf.\ Example~\ref{example:complement}).
This is the the set of existential configurations that cannot reach an accepting one via existential configurations only. 
Here, we first need to generate intervals encoding all (exponentially-sized) configurations of $\tm$, which can easily be achieved with cycle-free rules as described in the proof of Theorem~\ref{thm:inc_free_infinite}.

Next, we have a rule that turns an interval
\[
(\forall, \cleft_0, q_0, \cright_0^R, \cleft_1, q_1, \cright_1^R)
\]
and an interval~$
(\overline{L_0}, \cleft_2, q_2, \cright_2^R)
$
such that $\cleft_2 q_2 \cright_2$ is a successor of $\cleft_1 q_1 \cright_1$
into the interval
$
(\overline{L_1}, \cleft_0, q_0, \cright_0^R)$.
Thus, the intervals with the label~$\overline{L_1}$ encode those universal configurations that have a  branch that ends in an existential configuration from which no accepting configuration is reachable.
These are exactly the configurations where no accepting subrun with one alternation starts. 

Hence, another exclusive rule yields the intervals~$
(L_1, \cleft_0, q_0, \cright_0^R)
$
such that 
$
(\overline{L_1}, \cleft_0, q_0, \cright_0^R)
$
has not been not generated so far, i.e., exactly the configurations where an accepting subrun with one alternation starts.  

Now, we are in a dual situation: We have a rule that turns an interval
\[
(\exists, \cleft_0, q_0, \cright_0^R, \cleft_1, q_1, \cright_1^R)
\]
and an interval~$
(L_1, \cleft_2, q_2, \cright_2^R)
$
such that $\cleft_2, q_2, \cright_2$ is a successor of $\cleft_1 q_1, \cright_1$
into the interval~$
(L_2, \cleft_0, q_0, \cright_0^R)$.
Thus, the intervals with the label~$L_2$ encode those existential configurations that have a  branch that ends in an universal configuration from which an accepting subrun with one alternation starts.
These are exactly the configurations where an accepting subrun with two alternation starts. 

We repeat this construction with labels~$L_{m'}$ for every $m' \le m$.
Finally, there is a rule producing the target interval~$\eta_T$ from the interval~$(L_m, \cleft, q, \cright^R)$
if $\cleft q \cright$ is the initial configuration.
 
Thus, the fixed point contains the interval~$(\eta_T)$ iff $\tm$ accepts $w$.
Note that all the rules we have constructed are indeed cycle-free (just consider the order of appearance in the description above). To conclude the proof, it remains to note that there are only polynomially many rules, which are all of polynomial size.

For the upper bound, we again use the algorithm described in the proof of Theorem~\ref{thm:excl_free_finite}, with the exponential upper bound on the runtime relying on the upper bound on the map values as in Theorem~\ref{thm:inc_free_infinite}, which holds here as well.
Hence, we obtain an alternating algorithm with exponential running time (due to the fact the we need to perform calculations on exponentially-sized binary numbers), but with only linearly many (at most $|D|$) alternations. 
\end{proof}

\section{Extended nfer}
\label{sec:extended}


There is a further generalization of the semantics of \nfer that permits cycles in the same specification as exclusive rules so long as the exclusive rules are not part of a cycle.
This relaxation of the restrictions on the use of exclusive rules is desirable to support more use-cases for the language.
Since it is possible to obtain deterministic semantics for such specifications, an extended version of \nfer could support this.
This change to the semantics is also interesting for our purposes, since the evaluation complexity of \infer with cycles is higher than the cycle-free evaluation complexity of \nfer with exclusive rules.


\subsection{Semantics}

To support cycles and exclusive rules in the same specification, we modify $\traceinterpf{\_}{}$ to compute the strongly-connected components of the directed graph $G(D)$ described in Section~\ref{sec:complexity}.
Recall that this graph has a vertex for each rule and edges from rules with an identifier on their left-hand side to rules with the same identifier on their right-hand side.
The strongly connected components of the graph represent rules included in cycles (\emph{non-trivial} components) or individual rules outside of cycles (\emph{trivial} components).
We interpret the rules from each component as a specification with cycles and ensure deterministic behavior from exclusive rules by evaluating the components in topological sort order.


\begin{lstlisting}
  TraceSemE[:_:] : DeltaList -> Pool -> Pool
  TraceSemE[:delta1 ... deltaN:]pi = piSIZEl :- pi1 = pi && Comps = SCC(delta1 ... deltaN) &&
                      (D1,,,Dl)  = topsort(Comps) :- piJ = TraceSemF[:Di:](piI)
\end{lstlisting}
where $\text{SCC}(\delta_1 \cdots \delta_n)$ is the set $\mathcal{D}$ of strongly connected components of the directed graph $G(\delta_1 \cdots \delta_n)$ described in Section~\ref{sec:complexity} and $\textit{topsort}(\mathcal{D})$ is a topological sort of these components.
Here, $\traceinterpf{\_}{}$ is the interpretation function for full \nfer defined in Section~\ref{sec:exclusive}.
The results of $\traceinterpe{\_}{}$ are independent of the topological sort, as any such ordering will guarantee that all intervals matched by a sub-specification exist before it is evaluated using $\traceinterpf{\_}{}$.
Both the strongly-connected components and their topological sort can be computed in linear time~\cite{tarjan1972graph}.

\begin{example}
\label{example:primes}
We can now extend Example~\ref{example:evenodd} with inclusive and exclusive rules to compute the composite and prime numbers.
Recall the five \infer rules from that example to compute the natural (labeled with the identifier $N$), even ($E$), and odd ($O$) numbers from the input trace $\tau = (I,\set{v \mapsto 0})$.
We introduce two new rules that first compute the composite ($C$) numbers and then use them to find the primes ($P$):
\begin{align*}
\scalebox{.9}{$\delta_6 =$} &\ \scalebox{.9}{$\coinciderule{C}{N}{N}{\!{m_1,m_2 \mapsto m_1(v) > 1 \wedge m_2(v) > 1\!}}{\!{m_1,m_2 \mapsto \{ v \mapsto m_1(v) \cdot m_2(v) \}\!}}$}\\
\scalebox{.9}{$\delta_7 =$} &\ \scalebox{.9}{$\nferrule{P}{N}{\unlesskw\ \containkw}{C}{\!{m_1,m_2 \mapsto m_1(v) = m_2(v)\!}}{\!{m_1 \mapsto \{ v \mapsto m_1(v) \}\!}}$}.\\
\end{align*}
Rule $\delta_6$ matches two $N$-labeled intervals with map values greater than one and creates a $C$-labeled interval with its map value set to the product of the values from the two matched intervals.
Rule $\delta_7$ then finds the complement of the composite numbers defined in $\delta_6$ in a similar way to how an exclusive rule was used in Example~\ref{example:complement} to define odd numbers as the complement of even numbers.

Figure~\ref{fig:primes} shows the directed graph formed by \nfer rules $\delta_1 \cdots \delta_5$ from Example~\ref{example:evenodd} with $\delta_6$ and $\delta_7$ added.
In the graph, there are two non-trivial strongly-connected components: the cycles formed by $\delta_4$ and $\delta_5$ and the cycle formed by $\delta_3$ with itself.
The exclusive rule $\delta_7$ (shown as a square shape in the figure) is not included in these components and must only fall after after $\delta_1,\delta_3,$ and $\delta_6$ in any topological sort.
\begin{figure}[h]
  \centering
  \tikzset{>=stealth}
  \begin{tikzpicture}[shorten >=1pt,node distance=2cm,on grid,auto,ultra thick] 
   \node[state] (r1)  {$\delta_1$}; 
   \node[state] (r2)  [right=8cm of r1] {$\delta_2$};
   \node[state] (r3)  [right=3.5cm of r1] {$\delta_3$};
   \node[state] (r4)  [above right=1.5cm and 6cm of r1] {$\delta_4$};
   \node[state] (r5)  [below right=1.5cm and 6cm of r1] {$\delta_5$};
   \node[state] (r6)  [above left=1.5cm and 3cm of r1] {$\delta_6$};
   \node[rectangle,draw,minimum width=1cm,minimum height=1cm] (r7)  [below left=1.5cm and 3cm of r1] {$\delta_7$};
    \path[->] 
    (r1) edge node[above=0] {$N$} (r3)
         edge [bend left=30] node[above=0] {$N$} (r4)
         edge [bend right=30] node[above=0] {$N$} (r5)
         edge node[below] {$N$} (r6)
         edge node[above] {$N$} (r7)
    (r2) edge node[below=0] {$E$} (r5)
    (r3) edge [loop above=20] node[above=0] {$N$} ()
         edge node[above=0] {$N$} (r4)
         edge node[below=0] {$N$} (r5)
         edge [bend right=20] node[above] {$N$} (r6)
         edge [bend left=20] node[below] {$N$} (r7)
    (r4) edge [bend right] node {$E$} (r5)
    (r5) edge [bend right] node {$O$} (r4)
    (r6) edge node {$C$} (r7);
  \end{tikzpicture}
  \caption{The directed graph formed by the rules to compute composite and prime numbers}
  \label{fig:primes}
\end{figure}
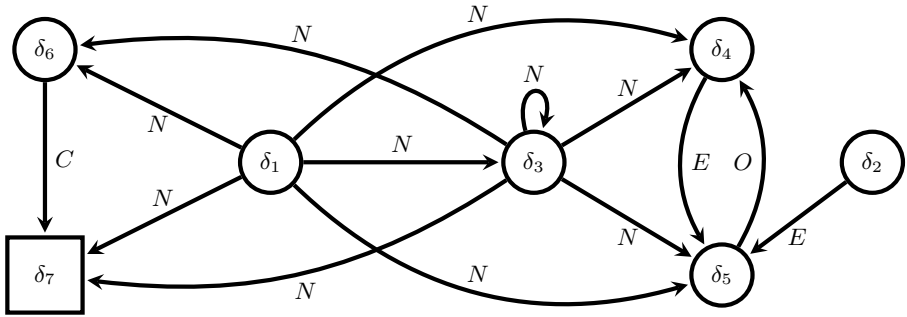
\end{example}


\section{Complexity Results for Extended nfer}

The \extnfer evaluation problem with infinite data is trivially undecidable, as already the \infer evaluation problem with infinite data is undecidable (see Theorem~\ref{thm:full}).
So, we only consider the case of finite data.

\begin{theorem}
\label{thm:extended_finite}
The \extnfer evaluation problem with finite data is $\exptime$-complete.
\end{theorem}

\begin{proof}
The lower bound already holds for the \infer evaluation problem with finite data (see Theorem~\ref{thm:inc_finite}), so we only need to prove the upper bound.

To this end, we again use Algorithm~\ref{algo_cyclesfinexcl}, but replace the bound~$|D|$ in Line~$3$ by $(|D|+1) \cdot b(D, \tau, k)$, where $b(D, \tau, k)$ is the upper bound on the number of intervals defined in the proof of Theorem~\ref{thm:inc_finite}.
Recall that this bound is also an upper bound on the height of a witness tree for an interval.

The correctness of the algorithm follows by combining the proofs of Theorem~\ref{thm:inc_finite} and Theorem~\ref{thm:excl_free_finite}: As exclusive rules do not appear on cycles, each one of them can be applied at most once when producing a fixed interval. 
Furthermore, between each such application no interval needs to be repeated in a witness tree. 
Altogether, the bound~$(|D|+1) \cdot b(D, \tau, k)$ is sufficient to determine whether a given interval is in the fixed point or not.

Finally, it is clear that the algorithm uses only polynomial space, as it only needs to store at most two intervals (and the resulting map values) and a counter bounded by $(|D|+1) \cdot b(D, \tau, k)$). 
All these objects can be encoded using polynomially many bits.
Thus, the upper bound on the complexity follows due to $\apspace = \exptime$.
\end{proof}

Finally, let us consider \extnfer combined with minimality. 
With finite data, the resulting evaluation problem is in $\ptime$ while it is in $\exptime$ with infinite data.
Both results can be proven using the same arguments as for \nfer (i.e., Theorem~\ref{thm:minimality_finite} and Theorem~\ref{thm:minimality_infinite}), as these rely on the number of possible intervals in the fixed point and the time required to compute their maps, but not the types of rules considered.

\section{Minimality}
\label{sec:minimality}
This section discusses the \emph{minimality} restriction and its implications on the complexity of the \nfer evaluation problem.
Traditionally, \nfer supports the concept of a \emph{selection function} that may modify the results of $\interp{\_}$~\cite{kauffman2017inferring}.
The reason is to support minimality, which filters any intervals that are not minimal in their timestamps.
Although minimality was originally introduced for its utility~\cite{kauffman2016nfer}, it has positive implications for evaluation complexity as well.
All results in this section hold for our extended semantics, i.e., with cycles and exclusive rules. 

Figure~\ref{fig:minimality} shows the effect of minimality on the evaluation of a single rule.
In the figure, time moves from left to right and events $B$ and $C$ on the timeline are the inputs to $\interp{\nferrule{A}{B}{\beforekw}{C}{\true}{\emptymap}}$.
This evaluation produces the three intervals labeled $A$ but minimality discards the longer, shaded interval because there are shorter $A$ intervals in the same period.

\begin{figure}[t]
  \centering
  \includegraphics[height=2cm]{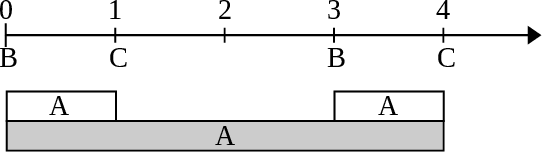}
  \caption{Minimality discards the shaded interval produced by $A \leftarrow \before{B}{C}$}
  \label{fig:minimality}
\end{figure}


Given a pool $\pool$ of existing intervals and a pool $\pool'$ of intervals to add, 
the $\minimality$ function returns only the minimal intervals in $\pool'$ that do not subsume any interval in $\pool$.
That is, the intervals where there is not another interval with the same identifier with a shorter duration during the same time.
No new intervals will be produced with the same identifier and timestamps when one already exists in $\pool$.
If there are multiple intervals with the same identifier and the same timestamps in $\pool'$, the one with the least map is retained (with respect to some fixed ordering of maps).
We define minimality as the following:

$\minimality : \poolty \times \poolty \rightarrow \poolty$\\
\indent $\minimality (\pool', \pool) = $\\
\begin{tabular}{p{1em}p{15mm}p{18mm}p{24em}}
  & $\{ i \in \pool' :$ & $\nexists i_1 \in \pool \where$ & $\Idof{i} = \Idof{i_1} \wedge \Startof{i} \leq \Startof{i_1} \wedge \Endof{i_1} \leq \Endof{i} \}\ \cap$\\
  & $\{ i \in \pool' :$ & $\nexists i_2 \in \pool' \where$ & $\Idof{i} = \Idof{i_2} \wedge {(\Startof{i} \leq \Startof{i_2} \wedge \Endof{i_2} < \Endof{i})}$ $\vee\ (\Startof{i} < \Startof{i_2} \wedge \Endof{i_2} \leq \Endof{i}) \}\ \cap$\\
  & $\{ i \in \pool' :$ & $\nexists i_3 \in \pool' \where$ & $\Idof{i} = \Idof{i_3} \wedge {\Startof{i} = \Startof{i_3} \wedge \Endof{i} = \Endof{i_3}}$ $\wedge\ \Mapof{i_3} \prec \Mapof{i} \}$
\end{tabular}

\noindent
where $\prec$ is a total order over $\mapty$ used as a tiebreaker when more than one new intervals exist in $\pool'$ with equal identifiers and timestamps.


For the \nfer evaluation problem under minimality we replace $\interp{\_}$ in the semantics with an interpretation function that applies \minimality to the result of $\interp{\_}$.

\begin{lstlisting}
       RuleMin[:_:] : Delta -> Pool -> Pool
       RuleMin[:rule:]pi = minimality(RuleSem[:rule:]pi,pi)
\end{lstlisting}

\begin{example}
Assume an input trace of alternating $B$-labeled and $C$-labeled intervals with empty maps, say $\tau = (B,0,\emptymap),(C,1,\emptymap),(B,3,\emptymap),(C,4,\emptymap)$ and the \nfer rule $\delta = \nferrule{A}{B}{\beforekw}{C}{\true}{\emptymap}$.
These are the same conditions as shown in Figure~\ref{fig:minimality}, where the four input events are shown on the top line.
Note that this example differs from those shown throughout the rest of the paper because we need to include timestamps to show the effect of minimality.
We project $\tau$ to an initial pool (we omit writing the maps as they are not used) $\pi_1 = \{(B,0,0),(C,1,1),(B,3,3),(C,4,4)\}$.

Without minimality, $\interp{\delta}{(\pi_1)} = \pi_2 = \{(A,0,1), (A,0,4), (A,3,4) \}$, since there are three pairings of (the interval projection of) the events in $\tau$ that satisfy the clock predicate \beforekw{}, and the map predicate is always true.
Note that the start and end timestamps of the intervals in $\pi$ are taken from the input events that satisfy $\delta$.

We now examine the component parts of $\minimality(\pi_2,\pi_1)$.
There are three sets constructed in $\minimality$, and an interval must be in all three sets to be in the output.
\begin{enumerate}
  \item $\{ i \in \pi_2 : \nexists i_1 \in \pi_1 \where \Idof{i} = \Idof{i_1} \wedge \Startof{i} \leq \Startof{i_1} \wedge \Endof{i_1} \leq \Endof{i} \}$
  
  This set includes all of $\pi_2$ since no $A$-labeled intervals appear in $\pi_1$.
  
  \item $\{ i \in \pi_2 : \nexists i_2 \in \pi_2 \where \Idof{i} = \Idof{i_2} \wedge {(\Startof{i} \leq \Startof{i_2} \wedge \Endof{i_2} < \Endof{i})}\vee\ (\Startof{i} < \Startof{i_2} \wedge \Endof{i_2} \leq \Endof{i}) \}$
  
  This set includes only $\{(A,0,1), (A,3,4) \}$ because for $(A,0,4)$ the interval $(A,0,1)$ has the same identifier, the same start timestamp, and a strictly lower end timestamp.  It would also be excluded by $(A,3,4)$.
  
  \item $\{ i \in \pi_2 : \nexists i_3 \in \pi_2 \where \Idof{i} = \Idof{i_3} \wedge {\Startof{i} = \Startof{i_3} \wedge \Endof{i} = \Endof{i_3}} \wedge\ \Mapof{i_3} \prec \Mapof{i} \}$
  
  This set includes all of $\pi_2$ since there are no intervals in $\pi_2$ that vary only by their maps.
\end{enumerate}
As such, $\interpmin{\delta}{(\pi_1)} = \{(A,0,1), (A,3,4) \}$.
Thus, the minimality restriction filters out the non-minimal interval~$(A,0,4)$. 
\end{example}

\begin{theorem}
\label{thm:minimality_finite}
The \nfer evaluation problem with finite data and minimality is in \ptime.
\end{theorem}

\begin{proof}
Consider an instance with specification~$D$, trace~$\tau$, and bound~$k$ on the map values.
Due to minimality, the size of $\traceinterpf[k]{D}{\tau}$ is bounded by $(\iota\cdot t^2)+|\tau|$, where $\iota$ is the number of identifiers in $D$ and $\tau$ and $t$ is the number of timestamps in $\tau$. 
Note that this bound is independent of $k$.

Also, map values and timestamps can be represented with polynomially many bits in the size of $D$ and $\tau$.
Hence, we can compute $\traceinterpf[k]{D}{\tau}$ and check whether it contains an interval labeled by the target identifier in polynomial time.
\end{proof}

A similar approach works for infinite data.

\begin{theorem}
\label{thm:minimality_infinite}
The \nfer evaluation problem with infinite data and minimality is in \exptime.
\end{theorem}

\begin{proof}
Again, we can compute $\traceinterpf{D}{\tau}$ to determine whether it contains an interval labeled by the target identifier. 
Due to minimality, the size of $\traceinterpf{D}{\tau}$ is still bounded by $(\iota\cdot t^2)+|\tau|$.
However, map values can only be bounded doubly-exponentially (cf.~Theorem~\ref{thm:inc_free_infinite}), as one can apply a polynomial $\iota\cdot t^2$ times to the map values of the initial intervals.
Hence, intervals in $\traceinterpf{D}{\tau}$ can be represented with exponentially many bits in the size of the the specification and the trace, which yields an exponential running time of computing $\traceinterpf{D}{\tau}$.
\end{proof}

\section{Discussion and Conclusion}
\label{sec:discussion}
We have studied the complexity of the \nfer evaluation problem.
It is undecidable in the presence of recursion and infinite data, even without exclusive rules. 
In contrast, regardless of the presence of exclusive rules, the evaluation problem is decidable for cycle-free specifications or with respect to finite data.
Most importantly for applications, the problem is in \ptime if we impose the minimality constraint and restrict to finite data. 
While we only allow natural numbers and Booleans as map values, our upper bounds also hold for more complex data types, i.e., signed numbers, (fixed-precision) floating point numbers, and strings, which  were included in the original definitions~\cite{kauffman2016nfer,kauffman2017inferring}.

Most of our complexity bounds are tight, but we leave one gap.
%
The \nfer evaluation problem with infinite data and minimality is in \exptime while no nontrivial lower bounds are known.
The upper bound follows from the fact that the map values may be of doubly-exponential size, i.e., they require exponential time to be computed.
However, minimality is a very restrictive constraint that in particular severely limits the ability to simulate nondeterministic computations. 
Coupled with the fact that minimality implies a polynomial upper bound on the number of intervals in the fixed point, this explains the lack of a nontrivial lower bound.

All our lower bound proofs only use intervals with the same timestamps, i.e., the complexity stems from the manipulation of data instead of temporal reasoning. 
Similarly, the upper bound proofs are mostly concerned with encoding of data and the temporal reasoning is trivial.
One of the reasons is that \nfer rules do not create new timestamps for intervals; newly created intervals can only use timestamps that already appear in the input trace. 
This leaves only a polynomial number of combinations of start points and end points, which is (at least) exponentially smaller than the number of data values.
For this reason, we propose to investigate data-free \nfer to analyze the complexity of the evaluation problem with respect to the choice of temporal operators. 
In this case, there are only polynomially many possible intervals in the fixed point.
So, a trivial upper bound on the complexity is \ptime, but we expect better results for fragments. 

Finally, we are currently studying the \nfer satisfiability problem, i.e., given a specification~$D$ and a target identifier~$\eta$, is there a trace~$\tau$ such that $\traceinterpe{D}{\tau}$ contains an $\eta$-labeled interval?
This has a practical application in supporting \nfer users in checking the specifications for sanity: if no $\eta$-labeled interval can be produced from any input, then the specification probably does not capture the user's intent.

\subsection*{Acknowledgments}
\label{sec:ack}

This research was partly funded by the ERC Advanced Grant LASSO, the Villum Investigator Grant S4OS and DIREC, Digital Research Center Denmark.

\bibliographystyle{splncs04}
\bibliography{logic,sean-cv}

\end{document}